\documentclass[11pt,a4paper]{article}
\usepackage[utf8x]{inputenc}
\usepackage{amsmath,ifthen,color,eurosym}

\usepackage[]{datetime}
\usepackage{wrapfig}
\usepackage{aliascnt} 
\usepackage{hyperref}
\usepackage{latexsym,amsthm,amsmath,amssymb,url}
\usepackage{fancyhdr}
\usepackage{graphicx}
\usepackage{comment}
\usepackage{stmaryrd}
\usepackage{enumerate}
\usepackage{cite}
\usepackage{longtable}
\usepackage{tikz}
\tikzset{black node/.style={draw, circle, fill = black, minimum size = 5pt, inner sep = 0pt}}
\tikzset{white node/.style={draw, circle, fill = white, minimum size = 5pt, inner sep = 0pt}}
\tikzset{normal/.style = {draw=none, fill = none}}
\usepackage{ifthen}
\usepackage{array}

\usepackage{multirow}

\newcommand{\mynewtheorem}[2]{
  \newaliascnt{#1}{dummy}
  \newtheorem{#1}[#1]{#2}
  \aliascntresetthe{#1}
  \expandafter\def\csname #1autorefname\endcsname{#2}
}

\theoremstyle{plain}
  \mynewtheorem{theorem}{Theorem}
  \mynewtheorem{proposition}{Proposition}
  \mynewtheorem{corollary}{Corollary}
  \mynewtheorem{lemma}{Lemma}
\theoremstyle{remark}
  \mynewtheorem{conjecture}{Conjecture}
  \mynewtheorem{question}{Question}       
  \mynewtheorem{observation}{Observation}
  \mynewtheorem{claim}{Claim}

\hypersetup{
  pdftitle = {Recent techniques and results on the Erdös-Posa property},
  pdfauthor = {Jean-Florent Raymond and Dimitrios M. Thilikos},
  colorlinks = true,
  linkcolor = black!30!red,
  citecolor = black!30!green
}

\DeclareMathOperator{\genus}{\gamma}
\DeclareMathOperator{\cc}{\mathrm{\mathbf{cc}}}

\DeclareMathOperator{\tw}{\mathbf{tw}}
\DeclareMathOperator{\tpw}{\mathbf{tpw}}
\DeclareMathOperator{\tcw}{\mathbf{tcw}}
\DeclareMathOperator{\pw}{\mathbf{pw}}
\DeclareMathOperator{\polylog}{polylog}

\DeclareMathOperator{\desc}{desc}
\DeclareMathOperator{\lminor}{\mathbin{\leq_{\mathrm{m}}}}
\DeclareMathOperator{\ltminor}{\mathbin{\leq_{\mathrm{tm}}}}
\DeclareMathOperator{\limm}{\mathbin{\leq_{\mathrm{imm}}}}
\newcommand{\lleq}{\preceq}

\DeclareMathOperator{\xcover}{{\sf x-cover}}
\DeclareMathOperator{\xpack}{{\sf x-pack}}
\DeclareMathOperator{\vcover}{{\sf v-cover}}
\DeclareMathOperator{\vpack}{{\sf v-pack}}
\DeclareMathOperator{\ecover}{{\sf e-cover}}
\DeclareMathOperator{\epack}{{\sf e-pack}}

\newcommand{\sv}{{\sf v}}
\newcommand{\se}{{\sf e}}
\newcommand{\sx}{{\sf x}}
\newcommand{\ep}{{Erdős--Pósa}}
\newcommand{\N}{\mathbb{N}}
\newcommand{\R}{\mathbb{R}}
\newcommand{\floor}[1]{\left\lfloor#1\right\rfloor}
\newcommand{\intv}[2]{\left \{#1,\dots, #2 \right \}}
\newcommand{\itemref}[1]{\hyperref[#1]{(\ref*{#1})}}
\newcommand{\coref}[1]{\hyperref[#1]{Cor.\ \ref*{#1}}}

\date{\empty}

\title{Recent techniques and results on the\\  Erdős--Pósa property\thanks{Emails of the authors: ~ \href{mailto:jean-florent.raymond@mimuw.edu.pl}{\texttt{jean-florent.raymond@mimuw.edu.pl}}\ , \href{mailto:sedthilk@thilikos.info}{\texttt{sedthilk@thilikos.info}} .}}
\author{Jean-Florent Raymond\thanks{Institute of Computer Science of the University of Warsaw,
    Poland, and University of Montpellier, France. } \thanks{Supported by the
    grant PRELUDIUM
    2013/11/N/ST6/02706 of the Polish National Science Center (NCN).} \thanks{AlGCo project team, CNRS, LIRMM, Montpellier,
    France.}
\and Dimitrios  M. Thilikos$^\S$\thanks{Department of Mathematics, National
  and Kapodistrian University of Athens, Greece.}}

\begin{document}
\maketitle

\begin{abstract}
\noindent Several min-max relations in graph theory can be expressed in the framework of the \ep\ property.
Typically, this property reveals a connection between packing and covering problems on graphs. 
We describe some recent techniques for proving this property 
that are related to tree-like decompositions. We also provide an unified presentation of 
the current state of the art on this topic.
\end{abstract}

\medskip

\noindent{{\bf Keywords}: \ep\ property, min-max theorems, tree decompositions, tree partitions, girth, graph minors, topological minors, graph immersions.}

\section{Introduction}

A considerable part of combinatorics has been developed around  min-max theorems. 
Min-max theorems usually identify dualities between certain objects in graphs, 
hypergraphs, and other combinatorial structures. The target is to  prove 
that the absence of the primal object   
implies the presence of the dual one and vice versa.

A classic example of such a duality is Menger's theorem: the 
primal concept is the existence of $k$ internally disjoint paths between 
two vertex sets $S$ and $T$ of a graph $G$, while the 
dual concept is a collection of $k$ vertices
that intersect all $(S,T)$-paths. Another example is Kőnig's theorem
where the primal notion is the existence of a matching with $k$ edges 
in a bipartite graph and the dual one is the existence of a vertex cover of size $k.$
It is also known that, in case of general graphs, this duality 
becomes an approximate one, i.e., a vertex cover 
of size $2k.$
In both aforementioned examples, the duality 
relates the notions of  {\em packing} 
and {\em covering} of a collection ${\cal C}$
of combinatorial objects of a graph. In Menger's 
theorem  ${\cal C}$ consists of all $(S,T)$-paths 
of $G$ while in Kőnig's theorem  ${\cal C}$ is the 
set of all edges of $G.$
That way, both aforementioned min-max theorems can be stated, 
for some class of graphs ${\cal G}$ (called {\em host} class)
and some gap function $f:\N\rightarrow\N$,
as follows:
\begin{quote}
 {\sl For every graph $G$~in ${\cal G}$, either 
$G$ contains $k$-vertex disjoint objects in ${\cal C}$
or it contains $f(k)$ vertices intersecting all objects in ${\cal C}$ that appear in $G.$}
\end{quote}
Clearly, for the case of Menger's theorem the host class is the class of all graphs while 
in the case of Kőnig's theorem the host class is restricted to the class of bipartite 
graphs. In both cases the derived duality is an {\em exact} one in the sense that  $f$ is
the identity function. However, this is not 
the case if we want to extend the duality of  Kőnig's theorem 
in the case of all graphs, where we can consider $f:k\mapsto 2k$ 
(i.e., we have 
an {\em approximate} duality).

One of the most celebrated results about packing/covering dualities was obtained 
by Paul Erdős and Lajos Pósa in 1965 where the object to cover and pack was 
the set of all cycles of $G$~\cite{erdHos1965independent}.
In this case the host class contains all graphs, while $f:k\mapsto O(k\cdot\log k).$
Moreover, Erdős and Pósa proved that this gap is optimal in the sense that 
it cannot be improved to a function $f:k\mapsto o(k\cdot\log k).$
This result motivated a long line of research for min-max dualities that are 
exact or 
approximate. Since then, a multitude of results on the \ep\ property
have appeared for several combinatorial objects, including extensions 
to digraphs~\cite{lucchesi1978minimax, Seymour1996Packing, ReedRST96pack, havet:hal-00816135,Guenin2010packing}, rooted graphs~\cite{KakimuraKK12erdo,Pontecorvi20121134, 2014arXiv1411.6554J, 2014arXiv1412.2894B}, labeled graphs~\cite{Kawarabayashi2005271}, signed
graphs~\cite{HochstattlerNP06,Aracena1263580}, hypergraphs~\cite{Alon02cove,bousquet2013hitting,Bousquet15ecdi}, matroids~\cite{GeelenK09thee}, 
and other combinatorial structures~\cite{gyarfas1970helly}
(see~\cite{Reed97tree} for a survey on this topic). Also it is worth to stress that 
\ep\ dualities have been useful in more applied domains. 
For example, in bioinformatics, they were useful for upper-bounding
the number of fixed-points of a boolean 
networks~\cite{Aracena2008,Aracena1263580, Aracena2016RSnumber}.\medskip

The purpose of this paper is twofold. We first describe some recent techniques 
for proving \ep-type results, mainly based on techniques related to 
tree-like decompositions of graphs (\autoref{decompep}) and the parameter of girth (\autoref{fromgirth}). We focused our presentation
to  the description of general frameworks that, we believe, might be useful for 
further investigations. In \autoref{sec:cr}, we present negative results on classes defined by containment relations. Lastly, in \autoref{results}, we provide an extensive update 
of results on the \ep\ property, reflecting the current progress on this 
vibrant area of graph theory.

\section{Definitions}

Unless otherwise mentioned, graphs in this paper are finite, undirected, do
not have loops and they may have multiple edges. 
We call a graph {\em nontrivial} if it contains at least one edge.
We denote by $V(G)$ and $E(G)$ the
vertex and edge sets of a graph~$G$, respectively, and we set $|G| =
|V(G)|$ and $\|G\| = |E(G)|$ (counting multiplicities).
For every set $X$ of vertices of a graph $G$, the subgraph of $G$ \emph{induced} by $X$, that we write $G[X]$, is the graph $\left (X, E\cap {X\choose 2}\right)$. For every set $X$ of vertices (resp.\ edges), we define $G\setminus X$ as the graph $G[V(G) \setminus X]$ (resp.\ $(V(G), E(G)\setminus X)$). The \emph{degree} of a vertex $v$ of a graph $G$, that we write $\deg_G(v)$ is the number of vertices adjacent to $v$ in~$G$. We drop the subscript when there is no ambiguity.

For $\sx \in \{ \sv,
\se\}$, and $G$ a graph, let $A_\sx(G) = V(G)$ if $\sx=\sv$ and $A_\sx(G) =
E(G)$ if $\sx=\se.$ In this sense we use symbols ${\sf v}$ and ${\sf e}$
in order to distinguish the vertex and the edge variants of the properties/parameters  
that we are dealing with. A graph is {\em subcubic} its maximum degree is bounded by $3.$
For every $t\in \N$, we denote by $\theta_t$ the graph with two vertices and $t$~edges.

\paragraph{Local operations.}
The operation of {\em contracting an edge $\{x,y\}$} in a graph $G$ introduces a new vertex $v_{xy}$ and makes it adjacent 
with all neighbors of $x$ and $y$ and then deletes $x$ and $y.$
The operation of {\em lifting a pair of edges $\{x,y\}$, 
$\{y,z\}$} in a graph $G$
increases by one the multiplicity of the edge $\{x,z\}$ (or introduces this edge
if it does not exist)
and then reduces by one the  multiplicities of  $\{x,y\}$ and $\{y,x\}.$

\paragraph{Partial orders on graphs.}
Given two graphs $H$ and $G$, we say that $H$ is an {\em induced subgraph} 
of $G$ if $H$ can be obtained from $G$ after removing vertices. Additionally,  
$H$ is a {\em 
subgraph} of $G$ if it can be obtained by some induced subgraph of $G$ after removing edges.
We also say that $H$ is a {\em  minor} (resp. {\em topological minor}) of $G$ if it can be obtained by some subgraph of $G$ after contracting edges (after contracting edges with some endpoint of degree at most $2$). Finally, we say that a graph $H$ is an {\em immersion}
of a graph $G$ if it can be obtained from some subgraph of $G$
after lifting 
pairs of edges that share some common endpoint.

Given a graph $H$, we denote by ${\cal M}(H)$, ${\cal T}(H)$, ${\cal I}(H)$
the class of all graphs that contain $H$ as a minor, topological minor, or immersion respectively.


\paragraph{Packings and covers.}
Let $\mathcal{H}$ be a family of graphs and let $\sx \in \{\sv, \se\}.$
An {\em \sx-$\mathcal{H}$-cover} of $G$ is a set $C \subseteq A_\sx(G)$ such that
$G \setminus C$ does not contain any subgraph isomorphic to a member
of $\mathcal{H}.$
An {\em  \sx-$\mathcal{H}$-packing} in $G$ is a collection of $\sx$-disjoint subgraphs
of $G$, each being isomorphic to some graph of $\mathcal{H}.$

We denote by $\xpack_{\mathcal{H}}(G)$ the maximum size of an \sx-$\mathcal{H}$-packing
and by $\xcover_{\mathcal{H}}(G)$ the minimum size of an \sx-$\mathcal{H}$-cover in~$G.$
Clearly, by definition, it always hold that  $\xpack_{\mathcal{H}}(G)\leq \xcover_{\mathcal{H}}(G)$, for every graph~$G$.

\paragraph{The \ep\ property.}
Let $\mathcal{G}$ and $\mathcal{H}$ be two graph classes, and let $\sx
\in \{\sv,\se\}.$ We refer to ${\cal G}$ as the {\em host} graph class and by ${\cal H}$ 
as the {\em guest} graph class.
We say that
$\mathcal{H}$ has the \sx-\emph{\ep\
  property} for $\mathcal{G}$ if there is a function $f\colon \N\to\N$
such that the following holds:
\[
\forall G \in \mathcal{G},\ \xcover_{\mathcal{H}}(G) \leq f(\xpack_{\mathcal{H}}(G)).
\]
Any function $f$ satisfying the above inequality is called a
\emph{gap} of the \sx-\ep\ property of $\mathcal{H}$ for $\mathcal{G}.$
When a class of graphs has the \sx-\ep\ property for the class of finite
graphs, we simply say that it has the \sx-\ep\ property.

\paragraph{Rooted trees.}
A \emph{rooted tree} is a pair $(T,s)$ where $T$ is a tree and $s \in
V(T)$ is a vertex referred to as the \emph{root}. Given a vertex $x\in V(T)$, the
\emph{descendants} of $x$ in $(T,s)$,
denoted by $\desc_{(T,s)}(x)$, is the set containing each vertex
$w$ such that the unique path from $w$ to $s$ in $T$ contains $x.$
If $y$ is a descendant of $x$ and is adjacent to $x$, then it is a
\emph{child} of~$x.$

\paragraph{Tree partitions.}
A \emph{tree partition} of a graph $G$ is a pair ${\cal D}=({\cal X},T)$
where $T$ is a tree and  ${\cal X}=\{X_t\}_{t\in V(T)}$ is a
partition of $V(G)$ such that either $|T|=1$ or for every $\{x,y\}\in
E(G)$, there exists an edge $\{t,t'\}\in E(T)$ where
$\{x,y\}\subseteq X_{t}\cup X_{t'}.$  Given an edge $f=\{t,t'\}\in
E(T)$, we define $E_{f}$ as the set of edges with one endpoint in
$X_{t}$ and the other in $X_{t'}.$ The {\em width} of ${\cal D}$ is
defined as \[\max\{\max\{|X_{t}|\}_{t\in V(T)}, \max\{\|G[X_{t}]\|\}_{t\in V(T)},\max\{|E_{f}|\}_{f\in E(T)}\}.\]

The \emph{tree partition width} of $G$ is 
the minimum width over all  tree partitions of $G$ and will be denoted by $\tpw(G).$ 
Tree partitions have been introduced in~\cite{Seese} (see also \cite{Halin1991203}) and 
tree partition width has been defined for simple graphs in~\cite{Ding1996tree}.
The extension of this definition for multigraphs is due to~\cite{Chatzidimitriou2015logopt}.

A \emph{rooted tree partition} of a graph $G$ is a 
triple ${\cal D}=({\cal X},T,s)$ where $(T,s)$ is a rooted tree and 
$({\cal X},T)$ is a tree partition of $G.$

\paragraph{Tree decompositions.}
A \emph{tree decomposition} of a graph~$G$
is a pair~$(T,\mathcal{X})$, where $T$ is a tree and
$\mathcal{X}$ is a family $\{X_t\}_{t \in V(T)}$ of
subsets of $V(G)$ (called \emph{bags}) indexed by
elements of $V(T)$, such that the following holds
 \begin{enumerate}[(i)]
 \item $\bigcup_{t \in V(T)} X_t = V(G)$;
 \item for every edge~$e$ of~$G$ there is an element of~$\mathcal{X}$
containing both ends of~$e$;
 \item for every~$v \in V(G)$, the subgraph of~${T}$
induced by $\{t \in V(T),\ {v \in X_t}\}$ is connected.
 \end{enumerate}

The \emph{width} of a tree decomposition~${(T, \mathcal{X})}$ is defined as
equal to $\max_{t \in V(T)} {|X_t| - 1}.$ The
\emph{treewidth} of~$G$, written~$\tw(G)$, is the minimum width of any
of its tree decompositions.

\section{The  \texorpdfstring{Erdős--Pósa}{Erdös-Posa} property from graph decompositions}
\label{decompep}
%
%

Let ${\cal H}$ be a graph class, ${\bf p}$ be a graph parameter, and ${\sf x}\in\{{\sf v},{\sf e}\}.$
We say that a function $f \colon \N_{\geq 0} \to \N_{\geq 0}$ is a {\em ceiling} for the  triple $({\bf p},{\cal H},{\sf x})$
if 
for every graph $G$,  
${\bf p}(G)\leq f(\xpack_{{\cal H}}(G)).$ Intuitively, there is a
ceiling for the triple $({\bf p},{\cal H},{\sf x})$ if a large value
of $\bf p$ on a graph forces a large $\sx$-packing of elements of $\mathcal{H}.$

Given a graph parameter ${\bf p}$ and an integer $k$,
we denote $${\cal G}_{{\bf p}\leq k}=\{G,\ {\bf p}(G)\leq k\}.$$

\begin{theorem}\label{ceil}
Let ${\cal H}$ be a class of  graphs, ${\sf x}\in\{{\sf v},{\sf e}\}$,
${\bf p}$ be a graph parameter, let $f\colon\N \to \N$ be a
function and let $h_{r} \colon \N \to \N$ be a function, for every $r
\in \N.$ Suppose that the following two conditions hold:
\begin{enumerate}[\bf A.]
\item $f$ is a ceiling for the triple $({\bf p},{\cal H},{\sf x})$;
\item for every $r\in\N$, ${\cal H}$ has the ${\sf x}$-\ep\ property 
for ${\cal G}_{{\bf p}\leq r}$ with gap $h_{r}$;
\end{enumerate}
then ${\cal H}$ has  the ${\sf x}$-\ep\ property
with gap $k \mapsto h_{f(k)}(k).$
\end{theorem}

\begin{proof}
Let $G$ be a graph and let $k = \xpack_{\mathcal H}(G).$ We have ${\bf
  p}(G) \leq f(k)$, by definition of a ceiling. Therefore, $G\in {\cal
  G}_{{\bf p}\leq f(k)}$, and thus $\xcover_{\mathcal{H}}(G) \leq
h_{f(k)}(k).$
\end{proof}

\autoref{ceil} will be used as a master theorem for the results of this section. 

\subsection{Vertex version and tree decompositions}

In a breakthrough paper~\cite{ChekuriC13Larg}, Chekuri and Chuzhoy
proved that every graph of large treewidth can be partitioned into
several subgraphs of large treewidth, with a polynomial dependency
between the treewidth of the original graph, the one of the subgraphs,
and the number of subgraphs. In particular they proved the next result.

\begin{theorem}[\!\!{\cite[Theorem~1.1]{ChekuriC13Larg}}]\label{cc1}
  Let $G$ be a graph with $\tw(G) = k$, and let $h,r$ be two integers
  with $hr^2 \leq k/\polylog k$. Then there is a partition $G_1,
  \dots, G_h$ of $G$ into vertex-disjoint subgraphs such that
  $\tw(G_i) \geq r$ for every $i \in \intv{1}{h}$.
\end{theorem}

Besides, the results in \cite{ChekuriC13poly} provide a polynomial
bound for the grid exclusion theorem.
The \emph{$(p\times q)$-grid} (for $p,q \in \N$) is the graph with vertex set
$\intv{1}{p} \times \intv{1}{q}$ and edge set $\{\{(i,j),(i',j')\},\ |i-i'| + |j-j'|=1\}$.
\begin{theorem}[\!\!{\cite[Theorem 1.1]{ChekuriC13poly}}]\label{cc2}
  There is constant $\delta$ such that every graph of treewidth $k$
  contains as a minor a $(\Omega(k^\delta)\times \Omega(k^\delta))$-grid.
\end{theorem}

As every planar graph $G$ is a minor of the every $(p \times
 p)$-grid for $p=|G| + 2\|G\|$ (\cite[1.5]{robertson1994quickly}), these two results can be combined to give the following polynomial ceiling for planar
graphs.
\begin{corollary}[see also the proof of {\cite[Theorem 5.4]{ChekuriC13Larg}}]\label{chekceil}
There is a function $f_h(k)=h^{O(1)}\cdot k\cdot (\log k)^{O(1)}$ such that, 
for every planar graph $H$ on $h$ edges, $f_h$ is a ceiling for the triple $(\tw,{\cal M}(H),{\sf v})$.
\end{corollary}
Indeed, according to \autoref{cc1}, every graph of \emph{large} treewidth can be
partitioned into \emph{many} disjoint
subgraphs each with treewidth \emph{large enough} (i.e. polynomial,
according to \autoref{cc2}) to force a \emph{large} grid as a minor,
which in turn contains the desired planar graph.

A function $f \colon \R \to \R$ is said to be {\em superadditive} if $f(x) +
f(y) \leq f(x+y)$ for every pair $x,y$ of positive reals.
The following argument has been first used in \cite{FominST11stre}
(see also \cite{ChekuriC13Larg, Raymond2013edge,
  Chatzidimitriou2015logopt}).

\begin{lemma}\label{connep}
Let $\mathcal{H}$ be a family of  connected graphs. If $f$ is a superadditive ceiling for
$(\tw, \mathcal{H}, \sv)$ then $\mathcal{H}$ has the ${\sf v}$-\ep\ 
property with gap~$k \mapsto 6\cdot f(k) \log(k+1).$
\end{lemma}

\begin{proof}
Let us show the following for every integer $k$: for every graph $G$, if $\vpack_{\mathcal{H}}(G) = k$ then $\vcover_{\mathcal{H}} \leq 6f(k) \log(k+1).$
The proof is by induction on $k.$
The base case $k=0$ is trivial. Let $k>0$, and let us assume that the above statement holds for every non-negative integer~$k' <k$ (induction hypothesis).

Let $G$ be a graph such that~$\vpack_{\mathcal{H}}(G) = k.$
A \emph{separation} of $G$ of order $p\in \N$ is a pair $(A,B)$ of subsets of
$V(G)$ such that $G$ has no edge with the one endpoint in $A\setminus B$ and the
other one in $B\setminus A$, and $|A\cap B|=p$.
We will rely on the following claim.

\begin{claim}\label{sep}
  There is a separation $(A,B)$ of order at most $\tw(G)+1$ of $G$ such that

  \begin{align*}
    \vpack_{\mathcal{H}}(G[A\setminus B]) &\leq 2k/3\quad \text{and}\\
    \vpack_{\mathcal{H}}(G[B\setminus A]) &\leq 2k/3.
  \end{align*}
  
\end{claim}
\begin{proof}
We consider a special type of tree decomposition 
 called {\em nice tree decomposition}.
    A triple $(T,r,\{X_t\}_{t \in V(T)})$ is said to be a \emph{nice} tree
  decomposition of a graph $G$ if $(T, \{X_t\}_{t \in V(T)})$ is a
  tree-decomposition where the following holds:
  \begin{enumerate}
  \item every vertex of $T$ has degree at most 3;
  \item $(T,r)$ is a rooted tree and the bag of the root $r$ is empty ($X_r = \emptyset$);
  \item every vertex $t$ of $T$ is
    \begin{itemize}
    \item either a \emph{base node}, i.e.~a leaf of $T$ whose
      bag is empty ($X_t = \emptyset$) and different from the root;

    \item or an \emph{introduce node}, i.e.~a vertex with only one
      child $t'$ such that $X_{t} = X_{t'} \cup \{u\}$ for some $u \in
      V(G)$;

    \item or a \emph{forget node}, i.e.~a vertex with only one
      child $t'$ such that $X_{t} = X_{t'} \setminus \{u\}$ for some $u \in
      X_{t'}$;

    \item or a \emph{join node}, i.e.~a vertex with two children $t_1$ and
      $t_2$ such that $X_t = X_{t_1} = X_{t_2}.$
    \end{itemize}
  \end{enumerate}
It is known that every graph $G$
has a nice tree decomposition with width  
 $\tw(G)$~\cite{Kloks94}.
 We therefore can assume that 
 $(T,r, (X_t)_{t \in V(T)})$ is a nice tree decomposition  
 of $G$ of optimal width. For each $t \in V(T)$, we define
 \[G_t = G\left [ \bigcup_{s \in
      \desc_{(T,r)}(t)} X_{s}\right ]\mbox{ \quad and\quad }G_t^- = G_t \setminus X_t.\]
      
Let $t$ be a vertex of $T$ at minimal distance from a leaf subject to
the requirement $\vpack_{\mathcal{H}}(G_t^-) > 2k/3.$ Such a vertex
exists, as~$\vpack_{\mathcal{H}}(G_r^-) = \vpack_{\mathcal{H}}(G_r) = k.$
Observe that $t$ is either a forget node, or a join node. Indeed, for
every base node $u$ we have $\vpack_{\mathcal{H}}(G_{u}^-) =
0.$ Moreover, every introduce node $u$ with child $v$ satisfies
$\vpack_{\mathcal{H}}(G_{u}^-) =\vpack_{\mathcal{H}}(G_{v}^-)$,
since~$G_u^- = G_v^-.$\smallskip

\noindent \textit{First case:} $t$ is a forget node with child $u.$ We set $A = V(G_u)$ and $B = V(G) \setminus
V(G_u^-)$. Observe that $(A,B)$ is a separation and that we have $A \cap B = X_u$, therefore the order of $(A,B)$ is at most $\tw(G)+1$. 
If $k  = 1$, then  $\vpack_{\mathcal{H}}(G[A\setminus B]) = \vpack_{\mathcal{H}}(G_u^-) = 0$ (by definition of $t$), whereas the fact that $\vpack_{\mathcal{H}}(G[A]) = \vpack_{\mathcal{H}}(G)$ implies $\vpack_{\mathcal{H}}(G[B\setminus A]) = 0 \leq 2k/3$. When $k\geq 2$, we have the following inequalities:
\begin{align*}
 \vpack_{\mathcal{H}}(G[A\setminus B]) &  = \vpack_{\mathcal{H}}(G_{u}^-) \\
 &\geq \vpack_{\mathcal{H}}(G_{t}^-) -1 \quad \text{(as
                                                                    $t$
                                                                    is
                                                                    a
                                                                    forget
                                                                    node)}\\
  & >  \frac{2k}{3} -1\\
  & \geq \frac{k}{3}.
\end{align*}
When $k=2$, the last inequality follows from the fact that $\vpack_{\mathcal{H}}(G[A\setminus B])$ is an integer.
Notice that we always have \[\vpack_{\mathcal{H}}(G[A\setminus B]) + \vpack_{\mathcal{H}}(G[B\setminus A]) \leq k.\] Together with the above inequality, this implies that $\vpack_{\mathcal{H}}(G[B\setminus A]) \leq 2k/3$, wheras it follows from the definition of $t$ that $\vpack_{\mathcal{H}}(G[A\setminus B]) \leq 2k/3$.

\smallskip

\noindent \textit{Second case:} $t$ is a join node with children $u_1,
u_2.$ We set $A  = V(G_{u_i})$ and $B = V(G) \setminus
V(G_{u_i}^-)$, where $u_i$ is a child of $t$ such that
$\vpack_{\mathcal{H}}(G_{u_i}^-) \geq k/3.$ Such child exists because
$\vpack_{\mathcal{H}}(G_{t}^-) = \vpack_{\mathcal{H}}(G_{u_1}^-) +
\vpack_{\mathcal{H}}(G_{u_2}^-)$ (as $t$ is a join node) and
$\vpack_{\mathcal{H}}(G_{t}^-)>2k/3$, by definition of~$t$.
Here again, $(A,B)$ is a separation and its order is at most $\tw(G)+1$ given that $A \cap B = X_{u_i}$.
The inequality $\vpack_{\mathcal{H}}(G[A\setminus B]) \leq 2k/3$
follows from the definition of~$t$ and the choice of $i$ ensures that $\vpack_{\mathcal{H}}(G[A\setminus B]) \geq k/3$, hence $\vpack_{\mathcal{H}}(G[B \setminus A]) \leq 2k/3$, as above.
\end{proof}

Observe that  $\tw(G) \leq  f(k)$, by definition of $f.$
According to \autoref{sep}, there is a separation $(A,B)$ of order
at most $\tw(G)+1$ in~$G$ such that $k_A, k_B\leq \floor{2k/3}$, where $k_A = \vpack_\mathcal{H}(G[A\setminus B])$ and $k_B =
\vpack_\mathcal{H}(G[B\setminus A]).$
Moreover, since $(A,B)$ is a separation, there is no connected graph of $G\setminus (A\cap B)$ that have vertices in both $G[A \setminus B]$ and $G[B \setminus A]$.
Therefore, given that every graph of $\mathcal{H}$ is connected, we can construct a $\sv$-$\mathcal{H}$-cover of $G \setminus (A\cap B)$ by taking the union of a $\sv$-$\mathcal{H}$-cover of $G[A \setminus B]$ and of one of $G[B \setminus A]$.
In other words, we have
\begin{align*}
  \vcover_\mathcal{H}(G) &\leq \vcover_\mathcal{H}(G[A\setminus B]) +
   \vcover_\mathcal{H}(G[B\setminus A]) + |A \cap B|\\
                             & \leq\vcover_\mathcal{H}(G[A\setminus B]) +
  \vcover_\mathcal{H}(G[B\setminus A]) + f(k)+1\\
& \leq 6f(k_A)\log(k_A+1)+ 6f(k_B)\log(k_B+1) + f(k)+1.
\end{align*}
The last inequality above is obtained by applying the induction hypothesis on both $G[A \setminus B]$ and $G[B \setminus A]$.
Notice that in the case where $k=1$, we get $k_A = k_B = 0$ and we
have $\vcover_\mathcal{H}(G) \leq  f(k) \leq 6\cdot f(k) \log
(k+1).$ Therefore we now assume $k\geq 2.$
We can then deduce that $\frac{2k}{3} +1\leq \frac{7}{9}(k+1)$.

We then have:
\begin{align*}
\vcover_\mathcal{H}(G) & \leq 6\cdot (f(k_A) + f(k_B)) \log \left
                        (\frac{2k}{3} + 1 \right) + f(k)+1\\
                          & \leq 6\cdot f(k)\log \left (\frac{7(k+1)}{9} \right) + f(k)+1&\text{\hspace{-1.5cm}(superadditivity of $f$)}\\
                         & \leq 6\cdot f(k) \log (k+1) - 6\cdot
                           \log (9/7) f(k) + 2f(k)\\
                          & \leq 6\cdot  f(k)\log (k+1).  
  \end{align*}
The second inequality also requires that $f$ is monotone, which is the case because it is superadditive and it never takes negative values.
  \let\qed\relax
\end{proof}
\vspace{-9mm}
\hfill\mbox{$\square$}\\

From the fact that the function of \autoref{chekceil} is superadditive, we get the following consequence of \autoref{connep}.
\begin{corollary}[see also \cite{ChekuriC13Larg}
  and \cite{ChekuriC13poly}]\label{epconn}
There is a function $f_h(k) = h^{O(1)}\cdot k \polylog(k)$ such that, for every connected planar graph $H$ with $h$ edges, the 
class $\mathcal{M}(H)$ has the \ep\ property with gap~$f_h$.
\end{corollary}

Notice that the above proof strongly relies on the fact that $H$ is connected. The non-connected 
case requires some more ideas that are originating
from~\cite{RobertsonS86GMV} (also used for forests in~\cite{FioriniJW12excl}). We
expose them hereafter. We will need the two next lemmas.

\begin{lemma}[\cite{RobertsonS86GMV}]\label{sort}
  Let $q,k$ be two positive integers, let $T$ be a tree and let ${\cal
    A}_1,
  \dots, {\cal A}_q$ be families of subtrees of $T.$ Assume that for
  every $i\in \intv{1}{q}$, there are $kq$ elements of ${\cal A}_i$ that are
  pairwise vertex-disjoint. Then for
  every $i\in \intv{1}{q}$, there are $k$ elements $T_1^i, \dots,
  T_k^i$ of ${\cal A}_i$ such that
\[
  T^1_1, \dots T_k^1,  T^2_1, \dots T_k^2, \dots, T^q_1, \dots T_k^q
\]
are all pairwise vertex-disjoint.
\end{lemma}

The next lemma is the \ep\ property of subtrees of a tree. It
can be obtained from the fact that subtrees of a tree have the Helly property.
\begin{lemma}[see \cite{gyarfas1970helly}]\label{treeep}
  Let $T$ be a tree and let $\mathcal{A}$ be a collection of subtrees
  of $T.$ For every positive integer $k$, either $T$ has (at least) $k$
  vertex disjoint subtrees  that belong to $\mathcal{A}$, or $T$ has a
  subset $X$ of less than $k$ vertices such that no subtree of
  $T\setminus X$ belongs to $\mathcal{A}.$ 
\end{lemma}

We are now ready to deal with disconnected patterns.
\begin{lemma}[\cite{RobertsonS86GMV}]\label{ep-tree}
Let $w$ be a positive integer and let $H$ be a graph on $q$ connected
components. $\mathcal{M}(H)$ has the \sv-\ep\ property on the class of
graphs of treewidth at most~$w$ with gap $k \mapsto (w-1)(kq-1).$
\end{lemma}

\begin{proof}
Let $k$ be a positive integer. We want to show that either
$\vpack_{\mathcal{M}(H)}(G) \geq k$ or $\vcover_{\mathcal{M}(H)}(G)
\leq (w-1)(kq-1).$
Let $H_1, \dots, H_q$ be the connected components of~$H.$ Let $(T,
\mathcal{X})$ be a tree-decomposition of~$G$ of width~$w.$
For every subgraph $F$ of $G$, we denote by $T(F)$ the
subgraph of $T$ induced by the bags containing vertices of~$F.$ Notice
that $T(F)$ is connected if $F$ is connected.

For every $i\in \intv{1}{q}$, we let $\mathcal{H}_i$ be the class of
subgraphs of $G$ that are isomorphic to a graph in $\mathcal{M}(H_i)$
and we consider the class $\mathcal{T}_i =
\{T(F), F\in \mathcal{H}_i\}.$

If for every $i\in \intv{1}{q}$, $\mathcal{T}_i$ contains $kq$
vertex-disjoint trees, then according to~\autoref{sort} there is a collection
$\{T_i^j\}_{i \in \intv{1}{q},\ j \in
  \intv{1}{k}}$ of pairwise vertex-disjoint trees, with $T_i^j \in \mathcal{T}_i$ for every $i \in \intv{1}{q}$ and every $j \in
  \intv{1}{k}$.
Observe that for every two subgraphs $F,F'$ of $G$,
if $T(F)$ and $T(F')$ are vertex-disjoint, then so are $F$ and $F'.$
Therefore $G$ has a collection $\{F_i^j\}_{i \in \intv{1}{q},\ j \in
  \intv{1}{k}}$ of pairwise vertex-disjoint subgraphs such that $F_i^j$ is isomorphic to an element of
$\mathcal{H}_i$, for every $i \in \intv{1}{q}$ and $j \in
  \intv{1}{k}.$ Consequently, for every $j\in \intv{1}{k}$, $\bigcup_{i=1}^q F_i^j$ is a subgraph of $G$ containing a graph isomorphic to a member of $\mathcal{M}(H)$, and these subgraphs are vertex-disjoint for distinct values of~$j$.
This proves that in this case, $\vpack_{\mathcal{M}(H)}(G) \geq k.$

We therefore now assume that the above condition does not hold, namely
there is an index $i\in \intv{1}{q}$ such that $\mathcal{T}_i$ contains
less than $kq$ vertex-disjoint trees. \autoref{treeep} implies the
existence of a subset $X$ with $|X| \leq kq-1$ such that $T\setminus X$ is
free from subtrees isomorphic to a member of $\mathcal{T}_i.$
Let $Y$ denote the union of the bags indexed by vertices
in~$X.$ Observe that $|Y| \leq (w-1) |X| \leq (w-1)(kq-1).$
The choice of $Y$ ensures that $G\setminus Y$ has no subgraph isomorphic to a member of~$\mathcal{H}_i.$
Hence $\vcover_{\mathcal{M}(H_i)} \leq (w-1)(kq-1).$ We deduce $\vcover_{\mathcal{M}(H)} \leq (w-1)(kq-1).$
\end{proof}

\begin{corollary}\label{epdisc}
For every  planar graph $H$ with $h$ edges and $q$ connected components, the 
class  $\mathcal{M}(H)$ has the \ep\ 
property with gap~$k \mapsto q\cdot h^{O(1)}\cdot k^2\cdot \polylog(k).$
\end{corollary}

\subsection{Edge version and tree partitions}
%
The technique presented in the previous section to deal with hosts of
bounded treewidth cannot be straighforwardly translated to the setting
of the edge-\ep\ property. Indeed, in general, knowing that two vertex sets are
separated by a small number of vertices does not give any information
on the minimum number of edges separating these sets. For this reason, we consider alternative of treewidth that guarantees that small edge-separators can be found. However, to the best of our knowledge, it is not known whether the edge-\ep\ property always holds when the host graphs have bounded treewidth.


One possible edge-analogue of treewidth is tree partition width.
Recall that $\theta_t$ is the graph with two vertices and $t$~edges, for every $t\in \N$.
The following uses~\cite[Theorem~1.2]{Ding1996tree}.
\begin{lemma}\label{simple}
For every $t\in \N$, there exists a ceiling for the triple $(\tpw,{\cal M}(\theta_{t}),{\sf e}).$
\end{lemma}

\begin{proof}
According to \cite[Theorem~1.2]{Ding1996tree}, there is a function $f
\colon \N \to \N$ such that for every $p\in \N$, every simple graph $G$
satisfying $\tpw(G) \geq f(p)$ contains as a subgraph either a
$p$-wall, or a $p$-path, or a $p$-star, or a $p$-fan. We omit the
definition of these graphs here, but we note that each of them
contains a \se-$\mathcal{M}(\theta_t)$-packing of size $k$ as soon as $kt<p/2.$

Let $G$ be a graph such that $\tpw(G) \geq f(2kt)\cdot kt.$ If $G$ has a
multiedge $e$ of multiplicity $\geq kt$, then it clearly contains an 
\se-$\mathcal{M}(\theta_t)$-packing of size $k.$ Therefore we now
assume that all edges of $G$ have multiplicity less than~$kt$.
 Observe that, if we denote by $\underline{G}$ the underlying simple
 graph of $G$, we have $\tpw(\underline{G})\geq
 \frac{\tpw(G)}{kt}.$ Hence $\tpw(\underline{G}) \geq f(2kt)$
 and, by definition of $f$ and the remark above, $\underline{G}$
 contains a \se-$\mathcal{M}(\theta_t)$-packing of size $k.$
As $\underline{G}$ is a subgraph of $G$, the aforementioned packing
also belong to $G$, which proves the lemma.
\end{proof}

Let ${\cal H}$ be a class of  graphs. We define $\tilde{\cal
  H}$ as the set of all the {\em subgraph-minimal elements of ${\cal
    H}$}, i.e.,
\[\tilde{\cal H}=\{H,\ H\in{\cal H}\mbox{~and none
  of the  subgraphs of $H$ belongs to ${\cal H}$} \}.\]
  
   We define $\Delta({\cal H})$ as the maximum number of edges incident to a vertex 
   in a graph of  ${\cal H}$ (counting multiple edges). We also set  $\tilde{\Delta}({\cal H})=\Delta(\tilde{\cal H}).$

\begin{lemma}
  For every graph $H$ of $h$ edges, it holds that 
  $\tilde{\Delta}(\mathcal{M}(H))\leq h, \tilde{\Delta}(\mathcal{T}(H))\leq h, \tilde{\Delta}(\mathcal{I}(H))\leq 2h.$
\end{lemma}

\begin{lemma}\label{split}
Let ${\cal H}$ be a class of  connected non-trivial graphs where $\tilde{\Delta}({\cal H})\leq d.$
Then for every $r \in \N$, ${\cal H}$ 
has the ${\sf e}$-\ep\ property on  ${\cal G}_{{\bf tpw}\leq r}$ with
gap $g_{r}(k)=k \cdot r\cdot(dr + 1)$.
\end{lemma}

\begin{proof}
  Let $r \in \N.$ We will
  show the following for every $k\in \N$: for every graph $G \in {\cal
    G}_{{\bf tpw}\leq r}$, if $\epack_{\mathcal{H}}(G) = k$ then
  $\ecover_{\mathcal H}(G) \leq g_r(k)$.

  We proceed by induction. The base case $k=0$ is trivial. We thus
  assume that $k>0$ and that the above statement holds for every
  positive integer~$k'<k$ (induction hypothesis).

  Let $G \in {\cal G}_{{\bf tpw}\leq r}$ be a graph such that
  $\epack_{\mathcal{H}}(G) = k.$ We assume that $G$ is connected, as
  otherwise we can treat each connected component separately.
  
Let $(\{X_t\}_{t \in V(T)}, T,s)$ be an optimal tree partition
decomposition of~$G.$ We define $G_t = G\left [ \bigcup_{u \in
    \desc_{(T,s)}(t)} X_{u}\right ].$ For every edge $\{u,v\}$ of $T$ we denote by
$E_{\{u,v\}}$ the edges of $G$ with the one endpoint in $X_u$ and the
other one in~$X_v.$
Let $t$ be a vertex of $T$ of minimum distance from a leaf, subject to
$\epack_{\mathcal{H}}(G_t)>0.$

Let $M$ be a subgraph-minimal subgraph of $G_t$ isomorphic to some
member of $\mathcal{H}$ and let $t_1, \dots, t_p$ be the children of
$t$ such that $V(G_{t_i}) \cap V(M) \neq \emptyset$ for every $i \in
\intv{1}{p}$.
By minimality of $M$, it has no vertex with more than
$\tilde{\Delta}(\mathcal{H}) \leq d$ incident edges. As $|X_t|\leq r$, we
deduce that $p \leq rd$.

Let $C = E(X_t) \cup \bigcup_{i=1}^p E_{\{t, t_i\}}$. Notice that $|C| \leq r + dr^2$. Let us
consider then graph $G' = G \setminus C.$ Let $M'$ be a subgraph of $G'$ that is isomorphic to some
member of $\cal H.$
By minimality of $t$, $\epack_{\mathcal{H}}(G_{t_i})=0$, for every $i\in
\intv{1}{p}.$ Therefore, if $M'$ contained an edge $e \in E(G_{t_i})$
(for some $i\in
\intv{1}{p}$), it would also contain an edge of $E(G) \setminus
E(G_{t_i}).$ Since every graph of $\cal H$ is connected, $M'$ would
also need to contain some edge of $E_{\{t, u_i\}}$ in order to be connected to
edges of $E(G) \setminus
E(G_{t_i}).$
However $E(G') \cap E_{\{t, u_i\}} = \emptyset.$ We deduce that for
every subgraph $M'$ of $G'$ that is isomorphic to some member of $\cal
H$, we have $E(M') \cap E(M) = \emptyset.$ It follows that every
$\se$-$\mathcal{H}$-packing in $G'$ is edge-disjoint with $M.$

Hence $\epack_{\mathcal{H}}(G')<k$, as otherwise a packing of size $k$ in $G'$
would, together with $M$, yield a packing of size $k+1$ in $G$ whereas
$\epack_{\mathcal{H}}(G) = k.$
By applying the induction hypothesis on $G'$, there is a subset
$D\subseteq E(G')$ such that $\epack_{\mathcal{H}}(G'\setminus D) = 0$
and moreover~$|D| \leq g_r(k-1).$
It is easy to see that $C \cup D$ is an \se-$\mathcal{H}$-cover of
$G$.
Furthermore $|C \cup D| \leq r(dr+1) + g_r(k-1)
= g_r(k)$, as required.
\end{proof}

An application of \autoref{split} is the following result, which also
relies on \autoref{ceil} and~\autoref{simple}.
\begin{corollary}[see also \cite{Chatzidimitriou2015logopt}]\label{thetaEP}
For every $r \in \N_{\geq 1}$, $\mathcal{M}(\theta_{r})$ has the
$\se$-\ep~property.
\end{corollary}

However, according to the results in \cite{Ding1996tree}, the class of
graphs $H$ such that there is a ceiling for $(\tpw, \mathcal{M}(H),
\se)$ is rather limited.
An alternative counterpart to treewidth might be the tree-cut width. We do not provide the definition here, but we refer the reader to the article where this parameter has been introduced~\cite{Wollan15} (see also \cite{GiannopoulouPTRW2016} for an alternative definition).
The next result appeared in~\cite{Giannopoulou2016packing}
and is strongly based on the results of~\cite{Wollan15}.

\begin{theorem}
\label{giantsould}
For every planar subcubic graph $H$ with $h$ edges, there exists a
ceiling for the triple $(\tcw,{\cal I}(H),{\sf e}).$
\end{theorem}

The next Lemma is the counterpart of \autoref{split}, especially for the case of 
immersion models for graphs of bounded tree-cut width.

\begin{lemma}[\!\!\cite{Giannopoulou2016packing}]
\label{speimm}
Let $t$ be a positive integer and let $H$ be a connected non-trivial
planar subcubic graph of $h$ edges.
Then ${\cal I}(H)$ has the ${\sf e}$-\ep\ property on ${\cal G}_{{\bf
    tcw}\leq t}$ with gap $k \mapsto t^2h k.$
\end{lemma}

Using~\autoref{ceil}, \autoref{speimm}, and \autoref{giantsould} we can also derive the following.

\begin{corollary}[\!\!\cite{Giannopoulou2016packing}]\label{corrsplit}
Let $H$ be a connected non-trivial planar subcubic graph of $h$ edges.
Then ${\cal I}(H)$ has the ${\sf e}$-\ep\ property with gap $k \mapsto (hk)^{O(1)}.$
\end{corollary}

\section{The \texorpdfstring{Erdős--Pósa}{Erdös-Posa} property from girth}
\label{fromgirth}

In this section, we give another proof of the \ep\ Theorem that
highlights a technique for proving more general \ep-type results. The
technique can be informally summarized as follow. We prove that either
$G$ contains a \emph{small} cycle or that it can be reduced to a
smaller graph with the same packing and cover number.
We then apply induction on either the graph where a small cycle has
been deleted (in the first case), or on the reduced graph (in the
second case). This technique has been successfully
applied in \cite{FioriniJW12excl, Chatzidimitriou2015logopt}, for instance.

The girth of a graph is the minimum length of a
cycle in this graph. 
Let us first recall the following result.

\begin{lemma}[\!\!\cite{Thomassen1983129}, see also~\protect{\cite[Theorem 7.4.2]{Diestel05grap}}]\label{thom}
  There is a constant $c \in \R$, such that for every $q\in \N_{\geq
    1}$, every graph of minimum degree at least 3 and girth at least
  $c \log q$ contains $K_q$ as a minor.
\end{lemma}
 A direct consequence of this result is the following trichotomy.
\begin{corollary}\label{trich}
For every graph $G$ and every integer $q>1$, one of the following holds:
\begin{enumerate}[(i)]
\item $G$ has a cycle on at most $c\log q$ vertices;
\item $G$ has a vertex of degree at most 2;
\item $G$ contains $K_{q}$ as a minor,
\end{enumerate}
where $c$ is the constant of \autoref{thom}.
\end{corollary}

We now prove the lemma that implies the classic \ep\ 
Theorem both for the vertex and its edge version. 
Recall that $A_\sx(G)$ denotes $V(G)$ or $E(G)$, depending if
$\sx = \sv$ or~$\sx = \se.$
\begin{lemma}\label{chatzi}
For  every $q\in \N^{+}$ and every ${\sf x}\in\{{\sf v},{\sf e}\}$, 
the class ${\cal M}(\theta_{2})$  has the  \sx-\ep\ property 
 for the class of graphs excluding $K_{q}$ as a minor with gap $O(k\cdot \log q).$
\end{lemma}

\begin{proof}
We will prove that for every non-negative integer $k$ and every
$K_q$-minor-free graph $G$, either $G$ has $k$ $\sx$-disjoint cycles,
or $G$ has a subset $X\subseteq A_\sx(G)$ of size at most $c k \log q$ such that $G
  \setminus X$ is a forest, where $c$ is the constant of \autoref{thom}.
  We proceed by induction on the pair $(k, G)$, with the
  well-founded order defined by $(k', G') \leq (k, G) \iff ( k' \leq k\ \text{and}\
  |A_\sx(G')| \leq |A_\sx(G)|)$, for all graphs $G$, $G'$ and
  non-negative integers $k$, $k'.$
  
  The base cases corresponding to $k=0$ or $|A_\sx(G)|=0$ are
  trivial. Let us now assume that $k\geq 1$, $|A_\sx(G)| \geq 1 $, and
  that the lemma holds for every pair $(k', G')$ such that $(k', G')
  \leq (k, G).$
  
  According to \autoref{trich}, either $G$ has a cycle $C$ on at most
  $c\log q$ vertices, or it has a vertex $v$ of degree at most two, or it
  contains $K_{q}$ as a minor. The last case is not possible, as
  we require $G$ to be $K_q$-minor-free.

  Whenever the first case applies, we set $G' = G\setminus A_\sx(C)$ and we consider
  the pair~$(k-1, G').$ If $G'$ contains $k-1$ $\sx$-disjoint cycles,
  then $G$ contains $k$ $\sx$-disjoint cycles obtained by adding $C$
  to those of~$G'$ and we are done. Otherwise, the induction
  hypothesis implies the existence of a subset $X' \subseteq
  A_\sx(G')$ with $|X'| \leq  c(k-1)\log q$ such that $G' \setminus X'$
  is a forest. Then by definition of $C$, $X = X' \cup A_\sx(C)$ has size at most $c \log q$
  and $G\setminus X$ is a forest, as required.

  In the second case, we delete $v$ if it is isolated and we contract
  an edge $e$ incident with it otherwise. Notice that since we cannot
  apply the first case, this
  contraction does not decrease the maximum number of $\sx$-disjoint cycles in~$G.$ Also, we
  can assume without loss of generality that $v$ (respectively $e$) is
  not part of a minimum $\sx$-cover of cycles in $G$, as any vertex
  adjacent to $v$ (respectively edge incident with $e$) covers all the
  cycles covered by $v$ (respectively $e$). Therefore the obtained
  graph $G'$ satisfies $\xpack_{\mathcal{M}(\theta_2)}(G') = \xpack_{\mathcal{M}(\theta_2)}(G')$ and $\xcover_{\mathcal{M}(\theta_2)}(G') =
\xcover_{\mathcal{M}(\theta_2)}(G).$ It is not hard to see that $|A_\sx(G')| <
|A_\sx(G)|.$ Therefore we can apply the induction hypothesis on $G'$ and
obtain the desired result on $G'$, that immediately translates to
$G$ by the above remarks.
\end{proof}

By setting $q = 3k$ and observing that every graph containing $K_{3k}$
as a minor also contains $k$ vertex-disjoint cycles (hence also
edge-disjoint), \autoref{chatzi} yields the vertex and edge versions of the classic
\ep\ Theorem as a corollary.

The technique presented in this section has been used to show the
following results.

\begin{theorem}[\!\!\cite{FioriniJW12excl}]\label{fep}
  For every forest $H$, $\mathcal{M}(H)$ has the \sv-\ep\ property with
  gap~$O(k).$
\end{theorem}

\begin{theorem}[\!\!\cite{Chatzidimitriou2015logopt}, see also \cite{FioriniJS13}
  for the vertex case]\label{cep}
  For every positive integer $r$ and every $\sx \in \{\sv, \se\}$,
  $\mathcal{M}(\theta_r)$ has the \sx-\ep\ property with gap $O(k\log k).$
\end{theorem}

Actually, the ideas in~\cite{Chatzidimitriou2015logopt} permit us to replace ${\cal M}(\theta_{2})$
by ${\cal M}(\theta_{r}), r\geq 2$ in \autoref{chatzi}.

To extend the idea of~\autoref{chatzi} 
in order to prove that some graph class ${\cal H}$
has the \sx-\ep\ property with
  gap~$f\colon\N\to\N$,  one should  
show that 
 for every positive integer $k$ and every graph $G$ with
$\xpack_{\mathcal{H}}(G) \leq  k$,
\begin{itemize}
\item either there is a graph $G'$ with $\xpack_{\mathcal{H}}(G) =\xpack_{\mathcal{H}}(G')$ and $\xcover_{\mathcal{H}}(G) =
\xcover_{\mathcal{H}}(G')$ and such that  $|G'|+\|G'\|<|G|+\|G\|$ (reduction case);
\item or $G$ has a subgraph isomorphic to a member of $\mathcal{H}$ on
  at most $f(k)/k$ vertices/edges (progress case).
\end{itemize}

In both  proofs of \autoref{fep} and \autoref{cep}, the reduction
case is done using  the graph theoretic notion of a protrusion introduced in~\cite{BodlaenderFLPST09meta,BodlaenderFLPST09} (or variants of it). 
Roughly speaking, the idea is to identify large parts of the graph that 
have constant treewidth (or constant  tree partition width, in case of~\autoref{cep}) 
and a small interface towards the rest of the graph and then prove 
that they can be replaced  by smaller ones without changing 
the packing or the cover number.

\section{Results in terms of containment relations}
\label{sec:cr}

For every partial order $\lleq$ on graphs, and for every graph $H$,
let $$\mathcal{G}_{\lleq}(H)= \{G\mid\ H\lleq G\}.$$
For every $\sx \in \{\sv, \se\}$, we define 
\begin{eqnarray*}
\mathcal{EP}^{\sx}_{\lleq} & = & \{H\mid \mbox{$\mathcal{G}_{\lleq}(H)$ has the
\sx-\ep\ property}\}
\end{eqnarray*}
A general question on the \ep\ property is to characterize $\mathcal{EP}^{\sx}_{\lleq}$ 
for several containment relations.
In this section we mainly provide some negative results about this problem.
We start with the following
easy  observation.

\begin{lemma}
\label{sbgro}\label{sgr}
If $\lleq$ is the subgraph 
or the induced subgraph relation, $\sx \in \{\sv, \se\}$,
and $H$ is a non-trivial graph,  then
$\mathcal{G}_{\lleq}(H)$ has the 
\sx-\ep\ property, with  gap $f:k\mapsto k\cdot |A_{\sf x}(H)|.$
In other words, $\mathcal{EP}^{\sx}_{\lleq}$ is the set of all graphs.
\end{lemma}

%

\begin{proof}
Let $H$ and $G$ be two graphs and let $k =
\xpack_{\mathcal{G}_{\lleq}(H)}(G).$ Let $M_1, \dots, M_k$ be a
\sv-$\mathcal{G}_{\lleq}(H)$-packing (respectively \se-$\mathcal{G}_{\lleq}(H)$-packing) of size $k$ with the minimal
number of vertices (respectively edges). Observe that in this case,
$|M_i| = |H|$ (respectively $\|M_i\| = \|H\|$) for every $i \in \intv{1}{k}.$ Let $X = \bigcup_{i=1}^k
V(M_i)$ (respectively $X = \bigcup_{i=1}^k
E(M_i)$). As the packing we consider is of size
$k$, the graph $G \setminus X$ does not have
any subgraph isomorphic to a member of $\mathcal{G}_{\lleq}(H).$ Hence
$X$ is an \sv-$\mathcal{G}_{\lleq}(H)$-cover (respectively
\se-$\mathcal{G}_{\lleq}(H)$-cover), and besides we have $|X| = k
\cdot |H|$ (respectively $|X| = k
\cdot \|H\|$). \end{proof}

Notice that in case ${\sf x}={\sf v}$,  it is not necessary to demand that $H$ is non-trivial in the statement of  \autoref{sbgro}.

\subsection{Some negative results}
Let us now state several negative results on the \ep\ property of
classes related to topological minors.

In the proofs below, we use the notion of \emph{Euler genus} of a graph $G.$
The {\em Euler genus} of a non-orientable surface ${\rm {\rm \Sigma}}$
is equal to the non-orientable genus
$\tilde{g}({\rm {\rm \Sigma}})$ (or the crosscap number).
The {\em Euler genus}  of an orientable   surface
${\rm {\rm \Sigma}}$ is $2{g}({\rm {\rm \Sigma}})$, where ${g}({\rm {\rm \Sigma}})$ is  the orientable genus
of ${\rm {\rm \Sigma}}.$ 
We refer to the book of Mohar and Thomassen \cite{mohar2001graphs} for
more details  on graph embeddings.
The {\em Euler genus} $\gamma(G)$ of a graph $G$ 
 is the minimum Euler genus of a surface where $G$ can be embedded.

 \begin{lemma}\label{nega}
   Let $H$ be a non-planar graph. Then $\mathcal{T}(H)$ does not have the \sv-\ep\ property.
 \end{lemma}

 \begin{proof}
   Informally, we will construct, for every positive integer $k$, a graph $G_k$ by ``thickening'' the vertices and edges of~$H$. From the non-planarity of $H$ and the way this graph is constructed, we will deduce that $\vpack_{\mathcal{T}(H)}(G_k) = 1$. On the other hand, the connectivity provided by the thickening of $H$ will ensure that the removal of any $k-1$ vertices will leave at least one subdivision of $H$~unaltered.

   For every integers $k>0$ and $d$, we denote by $\Gamma_{d,k}$ the graph
obtained from a grid of width $dk$ and height $d+k-1$ by adding $k$
vertices $a_1, \dots, a_k$ (that we call \emph{apices}) and connecting $a_1$ to the $d$ first
vertices on the first row of the grid (starting from the left), $a_2$ to the $d$ next vertices,
and so on. For every $i \in \intv{0}{d-1}$, the set of vertices at indices
$\{ik+j,\ j\in \intv{0}{k-1}\}$ on the last row of $\Gamma_{d,k}$ is
called the \emph{$i$-th port} of $\Gamma_{d,k}.$ We will refer to the vertex
at index $ik+j$ of the last row as the $j$-th vertex of the $i$-th port.
 See \autoref{fig:gadj} for a drawing of~$\Gamma_{4,3}.$ On this
 drawing, the ports are $U_0, \dots, U_3.$

\newcommand{\vortex}[1]{
  \begin{scope}[#1]
    \foreach \a in {0,90,...,270}{
    \fill[green!70!red, opacity = 0.45,rounded corners, rotate = \a] (35:5.5) arc (35:-35:5.5) -- (-35:4.5) arc (-35:35:4.5) -- cycle;
  }

  \foreach \a in {0,30,...,330}{
    \foreach \i in {1,...,5}{
      \draw[rotate = -30] (\a:\i) node (N{\a}{\i}) {};
    }
  }
  \foreach \a in {45, 165, 285}{
    \draw[rotate = -30] (\a:0.5) node[white node] (A{\a}) {};    
  }
  %
  \foreach \a in {0,30,...,330}{
    \foreach \i in {2,...,5}{
      \pgfmathparse{int(\i-1)}
      \draw (N{\a}{\pgfmathresult}) -- (N{\a}{\i});
    }
  }
    \foreach \i in {1,...,5}{
    \foreach \a in {30,60,...,330}{
      \pgfmathparse{int(\a-30)}
      \draw (N{\pgfmathresult}{\i}) to[bend right = 12.5] (N{\a}{\i});
    }
  }
  \foreach \a in {0,30,...,90}{
    \draw (A{45}) -- (N{\a}{1});
  }
  \foreach \a in {120,150,...,210}{
    \draw (A{165}) -- (N{\a}{1});
  }
  \foreach \a in {240,270,...,330}{
    \draw (A{285}) -- (N{\a}{1});
  }
\end{scope}
}

\begin{figure}[ht]
  \centering
  \begin{tikzpicture}[every node/.style = black node]
    \vortex{}
    \draw (0:6) node[normal] {$U_0$}
    (90:6) node[normal] {$U_1$}
    (180:6) node[normal] {$U_2$}
    (-90:6) node[normal] {$U_3$};
    \draw (A{45}) +(45:-.25) node[normal] {$a_1$};
    \draw (A{165}) +(165:-.25) node[normal] {$a_2$};
    \draw (A{285}) +(285:-.25) node[normal] {$a_3$};
  \end{tikzpicture}
  \caption{The gadget $\Gamma_{4,3}$ used in \autoref{nega}.}
  \label{fig:gadj}
\end{figure}
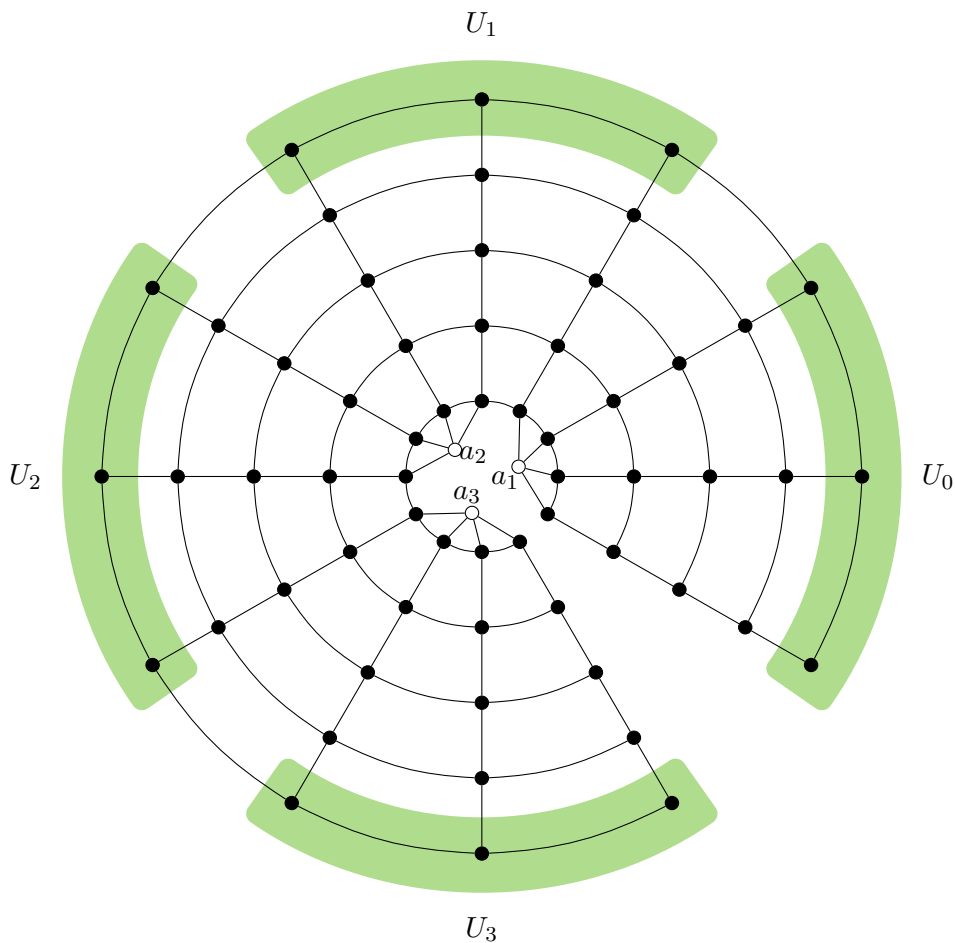

Let $k$ be a positive integer.
For every vertex $v$ of $H$, we arbitrarily choose an ordering of its
neighbors and we denote by $\sigma_v(u)$ the rank of $u$ in this ordering (ranging from 0 to $\deg(v)-1$), for every neighbor $u$ of~$v.$
We also let $F_v$ be a copy of the graph~$\Gamma_{\deg(v),k}.$

The graph $G_k$ can be constructed from the disjoint union of the
graphs of $\{F_v,\ v \in V(H)\}$ by adding, for every pair $u,v$ of adjacent
vertices, the edge connecting the $i$-th vertex of the $\sigma_v(u)$-th
port of $F_v$ to the $i$-th vertex of the $\sigma_u(v)$-th port of $F_u$,
for every $i \in \intv{0}{k-1}.$ Informally, we connect the vertices
of the $\sigma_v(u)$-th port of $F_v$ to the vertices
of the $\sigma_u(v)$-th port of $F_u$ using ``parallel'' edges.
\autoref{fig:k5gk} depicts the graph $G_k$ when $G = K_5$ and $k=3$. This graph contains a subdivision of $K_5$ but not two vertex-disjoint ones, and the removal of any two vertices leaves one subdivision of $K_5$ unaltered.

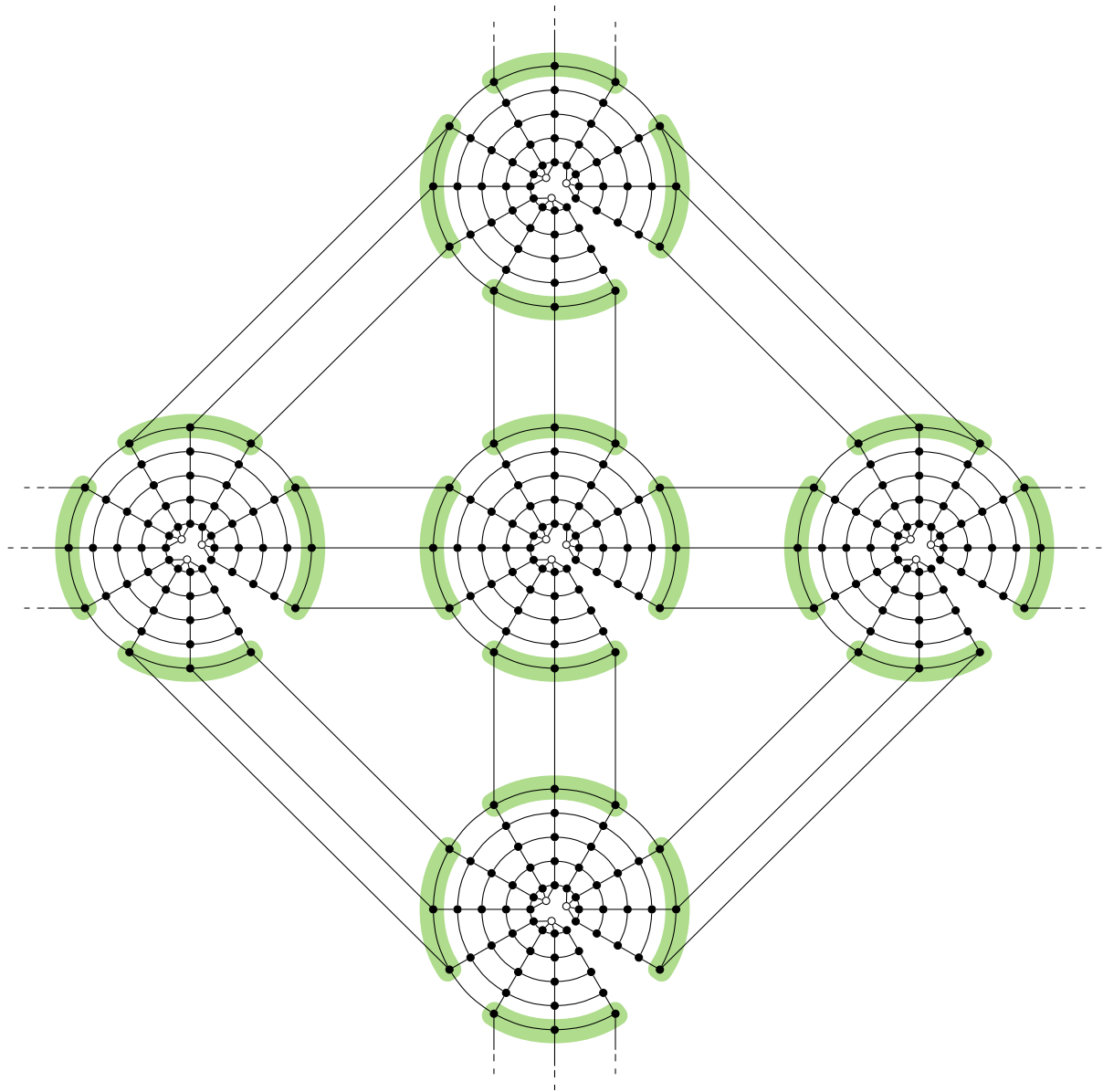
\begin{figure}[h!]
  \centering
  \begin{tikzpicture}[scale = 0.35,every node/.style = {black node, minimum size = 3pt}, white node/.style={draw, circle, fill = white, minimum size = 3pt, inner sep = 0pt}]
  \vortex{}
  \vortex{xshift = 15cm}
  \vortex{xshift = -15cm}
  \vortex{yshift = 15cm}
  \vortex{yshift = -15cm}
  
  \foreach \a in {-30, 0, 30}{
    \pgfmathparse{int(180-\a)}
    \draw (\pgfmathresult:5) -- ([xshift = -15cm]\a:5);
    \draw[xshift = 15cm] (\pgfmathresult:5) -- ([xshift = -15cm]\a:5);
  }
  \foreach \a in {60, 90, 120}{
    \draw (\a:5) -- ([yshift = 15cm]-\a:5);
    \draw[yshift = -15cm] (\a:5) -- ([yshift = 15cm]-\a:5);
  }
  \draw ([xshift = -15cm]60:5) -- ([yshift = 15cm]210:5);
  \draw ([xshift = -15cm]90:5) -- ([yshift = 15cm]180:5);
  \draw ([xshift = -15cm]120:5) -- ([yshift = 15cm]150:5);
  \draw ([yshift = -15cm]30:5) -- ([xshift = 15cm]240:5);
  \draw ([yshift = -15cm]0:5) -- ([xshift = 15cm]270:5);
  \draw ([yshift = -15cm]-30:5) -- ([xshift = 15cm]300:5);
  \draw ([yshift = -15cm]210:5) -- ([xshift = -15cm]240:5);
  \draw ([yshift = -15cm]180:5) -- ([xshift = -15cm]270:5);
  \draw ([yshift = -15cm]150:5) -- ([xshift = -15cm]300:5);
  \draw ([xshift = 15cm]60:5) -- ([yshift = 15cm]30:5);
  \draw ([xshift = 15cm]90:5) -- ([yshift = 15cm]0:5);
  \draw ([xshift = 15cm]120:5) -- ([yshift = 15cm]-30:5);
  %
  \foreach \a in {60, 90, 120}{
    \draw ([yshift = 15cm]\a:5) -- +(0,1.5) node[normal] (a){};
    \draw[dashed] (a) -- +(0,1);
  }
  \foreach \a in {240, 270, 300}{
    \draw ([yshift = -15cm]\a:5) -- +(0,-1.5) node[normal] (a){};
    \draw[dashed] (a) -- +(0,-1);
  }
  \foreach \a in {30, 0, -30}{
    \draw ([xshift = 15cm]\a:5) -- +(1.5,0) node[normal] (a){};
    \draw[dashed] (a) -- +(1,0);
  }
  \foreach \a in {150, 180, 210}{
    \draw ([xshift = -15cm]\a:5) -- +(-1.5,0) node[normal] (a){};
    \draw[dashed] (a) -- +(-1,0);
  }
\end{tikzpicture}
  \caption{The ``thickening'' of $K_5$ for $k=3$. Edges with dashed ends are connected to the aligned edges at the opposite side of the figure.}
  \label{fig:k5gk}
\end{figure}

It can be easily checked that the Euler genus of $G_k$ and $H$ are equal.
As $H$ is not planar, the Euler genus of the disjoint union of two copies of $H$ is
larger than the one of $H$ (see \cite{battle1962}) and
we get that~$\vpack_{\mathcal{T}(H)}(G) <2$. On the other hand, our
construction ensures that $\vpack_{\mathcal{T}(H)}(G) \geq 1.$

Let us now show that for every subset $X \subseteq V(G_k)$ with
$|X|<k$ we have $\vpack_{\mathcal{T}(H)}(G\setminus X) \geq 1.$
This would complete the proof, since $\{G_k,\ k \in \N_{\geq 1}\}$
would be an infinite family of graphs that have no
\sv-$\mathcal{T}(H)$-packings of size 2 but where a minimum
\sv-$\mathcal{T}(H)$-cover can be arbitrarily large.

Let $u$ and $v$ be two adjacent vertices of $H$, and let $d = \deg(v).$
For every $i \in
\intv{0}{k-1}$, let $C_i$ denote the vertices that are
\begin{itemize}
\item either in the same column of $F_u$ as the $i$-th vertex of the
  $\sigma_u(v)$-th port of $F_u$;
\item or in the same column of $F_v$ as the $i$-th vertex of the
  $\sigma_v(u)$-th port of $F_v.$
\end{itemize}
The family $\{C_i,\ i \in \intv{1}{k}\}$ contains $k$ vertex disjoint
elements, therefore at least one of them does not contain any vertex
from $X$ (as $|X|<k$).
Therefore, for every edge $\{u,v\}$ of $H$ there is an edge $f(\{u,v\})$ between a
vertex $x$ of the $\sigma_u(v)$-th port of $F_u$ and a vertex $y$ of
the $\sigma_v(u)$-th port of $F_v$ such that no vertex of the same column
as $x$ in $F_u$ (respectively $y$ in $F_v$) belong to $X.$
Using the same argument we can show that for every vertex $v\in V(H)$
there is an apex $a$ such that the columns of $F_v$ adjacent to $a$ are
free of vertices of~$X.$ Also we know that at least $d$ rows do
not contain vertices from $X$, as the grid of $F_v$ has height~$d+k-1.$
Therefore $F_v$ contains as a subgraph a grid $S_v$ such that:
\begin{enumerate}
\item some apex $a$ is adjacent to $d$ vertices of the first row of $S_v$;
\item for every edge $\{u,v\} \in E(H)$, the edge $f(\{u,v\})$ of $G_k$
  shares one vertex of the last row of $S_v$;
\item no vertex of the last row of $S_v$ belongs to two edges
  $f(\{u,v\})$ and $f(\{u',v\}) $ for some distinct neighbors $u,u'$
  of $v$;
\item $S_v$ has height and width at least~$d$;
\item $S_v$ does not contain any vertex of~$X.$
\end{enumerate}

We deduce that $F_v\setminus X$ contains $d$ paths $P_0^v, \dots,
P_{d-1}$ that have only the apex $a$ as common vertex and such that
$P_i$ connects $a$ to an endpoint of $f(\{v,u_i\})$, where $u_i$ is
the neighbor of $v$ of rank $i$, for every $i \in \intv{0}{d-1}.$
It is now easy to see that the graph
\[
G_k\left [ \bigcup_{v\in V(H)} \bigcup_{i=0}^{\deg_H(v)-1} V(P_i^v)\right]
\]
contains a subdivision of $H$ that does not contain any vertex
of~$X.$ This concludes the proof.
 \end{proof}

The proof of \autoref{nega} can be adapted to the setting of the
edge-\ep\ property under the additional requirement that the pattern is
subcubic.
\begin{lemma}\label{negae}
  Let $H$ be a subcubic non-planar graph. Then $\mathcal{T}(H)$ does not have the \se-\ep\ property.
\end{lemma}

\begin{proof}
  Let $k$ be a positive integer.
We use the same construction of $G_k$ as in the proof of
\autoref{nega} with the following modifications: each vertex $v$ of
degree $d\geq 4$ of $G_k$ is replaced by a subcubic tree, the leaves
of which are the neighbors of~$v.$ Let us call $G'_k$ the graph we
obtain.
It is not hard to see that the genus of $G'_k$ and $G_k$ are
equal. Moreover, as $G_k'$ is subcubic, every
\se-$\mathcal{T}(H)$-packing is also an \sv-$\mathcal{T}(H)$-packing. 
We then obtain as previously that $\epack_{\mathcal{T}(H)}(G'_k) =
1.$ The argument to show that $\ecover_{\mathcal{T}(H)}(G'_k) \geq k$
is identical to that used in the proof of \autoref{nega}.
\end{proof}

In fact, \autoref{nega} and \autoref{negae} can be used to prove that
more general classes do not have the \ep\ property, as follows. As we
will see in \autoref{neg-min} and \autoref{neg-min-e}, the
conditions of \autoref{famil} already encompass several well-studied classes.
\begin{lemma}\label{famil}
  Let $\sx\in\{\sv, \se\}$, let $H$ be a non-planar graph and let $\mathcal{H}$ be a class of graphs such
  that:
  \begin{enumerate}[(i)]
  \item $\mathcal{T}(H) \subseteq \mathcal{H}$; and\label{subH}
  \item $H$ is graph of minimum Euler genus in~$\mathcal{H}$;\label{minH}
  \item if $\sx=\se$, then $H$ is subcubic.
  \end{enumerate}
    Then $\mathcal{H}$ does not have the \sx-\ep\ property.
\end{lemma}

\begin{proof}
  Let $k$ be a positive integer.
  We again consider the graphs $G_k$ and $G_k'$ constructed from $H$ as in the proofs of
  \autoref{nega} and \autoref{negae}. Let $J_k$ be $G_k$ if $\sx =\sv$ and $J_k = G_k'$ if $\sx = \se.$ Let us show that $\vpack_{\mathcal{H}}(J_k) = 1.$ For this, let us
assume that there is an \sx-$\mathcal{H}$-packing $F_1, \dots, F_p$, for some $p\in
\N_{\geq 2}$ in $J_k.$ It is crucial to note that in both the cases
$\sx =\sv$ and $\sx =\se$, the subgraphs $F_1, \dots, F_p$ are
vertex-disjoint. In fact, when $\sx =\sv$, this follows from the definition of
a \sv-$\mathcal{H}$-packing, and if $\sx =\se$ it is because $G_k'$
is subcubic. Recall that $\gamma(G)$ denotes the Euler genus of $G$,
i.e.\  the minimum Euler genus of a surface where $G$ can be embedded, and that our construction ensures that $\gamma(J_k) = \gamma(H)$ (see the proofs of \autoref{nega} and \autoref{negae}).
Then we have:
\begin{align*}
  \genus(J_k) &\geq \genus(F_1 \cup
\dots \cup F_p) &\\
& = \sum_{i=1}^p \genus(F_i) & \text{(see \cite{battle1962})}\\
& \geq p \cdot \genus(H)&\text{(see below)}\\
\genus(J_k)& > \genus(H) & \text{(as $p\geq 2$)}.  
\end{align*}
The inequality $\genus(J_k) \geq p \cdot \genus(H)$ come from the
requirements \itemref{subH}-\itemref{minH}, which imply $\gamma(F_i)
\geq \gamma(H)$ for every $i\in \intv{1}{p}$. 
The last inequality contradicts the fact that $\gamma(J_k) = \gamma(H)$.
Therefore~$\vpack_{\mathcal{H}}(J_k) = 1.$
On the other hand, 
\begin{align*}
\vcover_{\mathcal{H}}(J_k) \geq \vcover_{\mathcal{T}(H)}(J_k)\geq k.
\end{align*}
The last inequality can be found in the proof of \autoref{nega} or
\autoref{negae} (depending if $\sx=\sv$ or $\sx = \se$).
This concludes the proof.
\end{proof}

 \begin{corollary}\label{neg-min}
   For every non-planar graph $H$, none of $\mathcal{I}(H)$ and $\mathcal{M}(H)$ have the \sv-\ep\ property.
 \end{corollary}

 \begin{corollary}\label{neg-min-e}
   For every subcubic non-planar graph $H$, none of $\mathcal{I}(H)$ and $\mathcal{M}(H)$ have the \se-\ep\ property.
 \end{corollary}

 \autoref{neg-min-e} can be strengthened by dropping the degree
 condition on $H$ when considering minor models of $H$, as follows.

 \begin{lemma}\label{neg-min-e-plus}
   For every non-planar graph $H$, $\mathcal{M}(H)$ does not have the \se-\ep\ property.
 \end{lemma}
 \begin{proof}
  Let $k$ be a positive integer. Again we use the graph $G_k'$
  constructed as in \autoref{negae}. We modify it by replacing every
  apex $a$ by a subcubic tree, the leaves of which are the neighbors
  of~$a.$ Let $G_k''$ denote the graph that we obtain.
  Observe that $G_k''$ is subcubic. Therefore, using the same argument as
  in the proof of \autoref{negae} we can show that
  $\epack_{\mathcal{M}(H)}(G) = 1.$
In the sequel we use the terminology of the proof of \autoref{nega}.
  Let $F_v''$ denote the graph obtained from $F_v$ by replacing every
  vertex $u$ of degree at least 4 by a subcubic tree, the leaves of
  which are the neighbors of $u$, for every $v \in V(H).$
  The proof that
  $\ecover_{\mathcal{M}(H)}(G) \geq k$ goes as in the proof of
  \autoref{nega}, except that we obtain, for every $v \in V(H)$, that
  $F''_v\setminus X$ contains a tree, the leaves of which are endpoints
  of $f(\{v, u_i\})$ for $i \in \intv{0}{d-1}$ (instead of paths
  connecting an apex to endpoints
  of $f(\{v, u_i\})$). Fortunately this is enough to guarantee that
  $G_k''\setminus X$ contains $H$ as a minor, and we are done.
 \end{proof}

Let us now summarize results related to the most common containment relations.
\begin{description}
\item[Subgraphs and induced subgraphs:] $\mathcal{EP}^\sx_{\lleq}$ is the class of
  all graphs, both for $\lleq$ being the subgraph and induced subgraph
  relation, for every $\sx \in \{\sv, \se\}$ (\autoref{sgr}).
\item[Minors:] $\mathcal{EP}^\sv_{\lminor}$ is the class of planar
  graphs~\cite{RobertsonS86GMV}. Recently, an extension of this  characterization, for strongly connected directed graphs, appeared in~\cite{AkhoondianKKW09thee}. 
  About the edge version, the authors
  of~\cite{Raymond2013edge} proved that $\mathcal{EP}^\se_{\lminor}$
  includes the class  $\{\theta_r\}_{r\in \N_{\geq 1}}$, and we show in
  \autoref{neg-min-e-plus} that $\mathcal{EP}^\se_{\lminor}$ is a subclass
  of planar graphs (see also \cite{Chatzidimitriou2015logopt}).
\item[Topological Minors:] $\mathcal{EP}^\sv_{\ltminor}$ has been
  characterized in~\cite{Liu2014topo}. There are trees that do not
  belong to
  $\mathcal{EP}^\sv_{\ltminor}$~\cite{thomassen1988presence}. The
  class $\mathcal{EP}^\sv_{\ltminor}$ does not contain any non-planar
  graph (\autoref{nega}) and $\mathcal{EP}^\se_{\ltminor}$ does not
  contain any non-planar subcubic graph (\autoref{negae}). Analogous characterizations
  for the case of strongly connected digraphs have recently appeared in~\cite{AkhoondianKKW09thee}.
  \item[Immersions:] As proved in~\cite{Giannopoulou2016packing}, 
   $\mathcal{EP}^\sv_{\limm}$ contains all planar
    subcubic graphs and $\mathcal{EP}^\se_{\limm}$ 
   contains all non-trivial, connected, planar
    subcubic graphs. 
   Moreover, $\mathcal{EP}^\sv_{\limm}$ does not contain any
   non-planar graph (\autoref{neg-min}) and $\mathcal{EP}^\se_{\limm}$ does not contain
   any subcubic non-planar graph (\autoref{neg-min-e}). On the other
   hand there is a 3-connected planar graph of maximum degree 4 that
   belongs to none of $\mathcal{EP}^\sv_{\limm}$ and
   $\mathcal{EP}^\sv_{\limm}$~\cite{Giannopoulou2016packing}.
 \end{description}

\section{Summary of results}
\label{results}

In the following sections we list positive and negative results on the
\ep\ property, and open problems.

Let us define the notation used in all the tables of
\autoref{posires} and \autoref{sec:neg}.
The fourth column of the tables gives the type of the
packings/covers the current line is about. The character \sv\ (respectively \se) refers to
vertex-disjoint (respectively edge-disjoint) packings and vertex
(respectively edge) covers. We write $\sv/\se$ when the mentioned
result holds for both the vertex and the edge version.
The symbol $\sv_{1/p}$ (resp. $\se_{1/p}$) for some $p\in \N$ indicates that the packing is allowed to use at most $p$ times each vertex (resp.\ each edge) and that the cover contains vertices (resp.\ edges).
 Finally, {\sf w} stands for vertex covers and packings where
every vertex $v$ of the host graph can be used at most $w(v)$ times by
every packing, where $w$ is a function mapping reals to the vertices
of the host graph.
The more specific definitions are given in the corresponding sections.

\subsection{Positive results}
\label{posires}

We provide a series of  tables presenting known results on the \ep\ property of
some graph classes, sorted depending on the pattern. Results related
to other structures (matroids, hypergraphs, geometry) and to
fractional versions are not mentioned here.

A dash in the ``gap'' column means that the authors did not
explicitly provided a gap function, even though one might be computable
from the proof.
The fourth column refers to the type of packing/cover, as defined above.

\subsubsection{Acyclic patterns}

Let $G$ be a graph. For every $S,T \subseteq V(G)$, an \emph{$(S,T)$-path} of $G$ is a path with the one endpoint in $S$ and the other one in~$T$.
An \emph{$S$-path} is a path with both endpoints (which are distinct) in $S$.
If $\mathcal{S}$ is a collection of subsets of $V(G)$, an \emph{$\mathcal{S}$-path} is a path that has endpoints in two different elements of~$\mathcal{S}$.
A generalization of these settings have been introduced in \cite{marx2015exact}, where the pairs of vertices that can be connected by a path are specified by an auxilliary graph. If $S \subseteq V(G)$ and $H$ (demand graph) is a graph with vertex set $S$, a path of $G$ is said to be \emph{$H$-valid} if its endpoints are adjacent vertices of~$H$.

\medskip

\noindent
\begin{tabular}{|p{2cm}|p{3.5cm}|p{3.5cm}|p{0.55cm}|p{3.95cm}|}
\hline
   Ref.& Guest class $\mathcal{H}$& Host class $\mathcal{G}$&
                                                              T. &Up.-bound on the gap \\ \hline\hline
   \cite{koniggrafok} & $K_2$ & bipartite & \sv & $k$\\\hline
\cite{lucchesi1978minimax} & \multirow{2}{*}{directed cuts} & \multirow{2}{*}{any digraph} &  \multirow{2}{*}{\se}& \multirow{2}{*}{$k$}\\
\cite{LOVASZ197696}&&&&\\
\hline
\cite{menger1927allgemeinen} & $(S,T)$-paths & any & \sv/\se & $k$\\\hline
\cite{grunwald1938neuer} & directed $(S,T)$-paths & any digraph & \sv/\se & $k$\\ \hline
\cite{Gallai1964} & $S$-paths & any & \sv & $2k$\\\hline
\cite{mader1978maximalzahl} & $\mathcal{S}$-paths & any & \sv & see \cite{schrijver2001short} \\ \hline
\cite{mader1978maximalzahlA} & $\mathcal{S}$-paths & any & \se & see \cite{Sebo2004} \\ \hline
  \cite{Chudnovsky2006packing} & non-zero directed $S$-paths & edge-group-labeled digraphs & \sv & $2k-2$\\ \hline
\cite{marx2015exact} & $H$-valid paths, $H$ with no matching of size $t$ & any & \sv & $2^{2^{O(k+t)}}$\\\hline
  \parbox{2cm}{\cite{FioriniJW12excl},\\ \hyperref[fep]{Th.\ \ref*{fep}}\vspace{1mm}} & $\mathcal{M}(H)$,\ $H$ forest & any & \sv & $O_{\mathcal{H}}(k)$\\ \hline
\end{tabular}

\subsubsection{Triangles}
A graph is \emph{flat} if every edge belongs to at most two triangles.
\medskip

\noindent
\begin{tabular}{|p{2cm}|p{3.5cm}|p{3.5cm}|p{0.55cm}|p{3.95cm}|}
\hline
   Ref.& Guest class $\mathcal{H}$& Host class $\mathcal{G}$&
                                                              T. &Up.-bound on the gap\\ \hline\hline
\multirow{3}{*}{\cite{Tuza90aconj}} & \multirow{3}{*}{triangles} & planar graphs &\se& $2k$
  \\ \cline{3-5}
&& $G$ with $||G|| \geq 7|G|^2/16$ & \se& $2k$\\ \cline{3-5}
&&tripartite graphs & \se& 7k/3\\ \hline
\cite{Krivelevich1995281} & triangles &$\mathcal{T}(K_{3,3})$-free graphs & \se &
                                                                       $2k$\\
  \hline
\cite{Haxell88pack} & triangles &tripartite graphs & \se & $1.956k$ \\ \hline
\cite{Haxell1999251} & triangles & any & \se & $(3-\frac{3}{23})k$\\
  \hline
\multirow{2}{*}{\cite{AparnaLakshmanan2011}} & \multirow{2}{*}{triangles} & odd-wheel-free graphs
                                                                 & \multirow{2}{*}{\se}
                                    & \multirow{2}{*}{$2k$}\\
  \cline{3-3}
&&4-colorable graphs &&\\ \hline
\multirow{2}{*}{\cite{Haxell2011}} & \multirow{2}{*}{triangles} & $K_4$-free planar graphs& \multirow{2}{*}{\se} & \multirow{2}{*}{$3k/2$}\\ \cline{3-3}
 & & $K_4$-free flat graphs& &\\ \hline
\end{tabular}

\subsubsection{Cycles}

The statement of the results in \cite{Ding2002381, Ding2003244}
requires additional definition.
An odd ring is a graph obtained from an odd cycle by replacing every
edge $\{u,v\}$ by either a triangle containing $\{u,v\}$, or by two triangles on vertices $\{u,a,b\}$ and
$\{v,c,d\}$ together with the edges $\{b,c\}$ and $\{a,d\}.$
We denote by $\mathcal{G}_1$ the class of graphs with no induced
subdivision of the following: $K_{2,3}$, a wheel, or an odd ring.
We denote by $\mathcal{G}_2$ the class of graphs with no induced
subdivision of the following: $K_{3,3}$, a wheel, or an odd~ring.

The results on directed cycles also need few more definitions. A digraph is \emph{strongly planar} if it has a planar drawing such
that for every vertex $v$, the edges with head $v$ form an interval in
the cyclic ordering of edges incident with $v$ (definition
from~\cite{Guenin2010packing}). An \emph{odd double circuit} is a
digraph obtained from an undirected circuit of odd length more than 2
by replacing each edge by a pair of directed edges, one in each
direction. $F_7$ is the digraph obtained from the directed cycle on
vertices $v_1, \dots, v_7,v_1$, by adding the edges creating the directed cycle
$v_1,v_3, v_5, v_7, v_2, v_4, v_6, v_1.$ We denote by $\mathcal{F}$ the
class of digraphs with no butterfly minor isomorphic to an odd double circuit, or~$F_7$
(for the definition of butterfly minors of digraphs see~\cite{Guenin2010packing,JohnsonRST01dire,AkhoondianKKW09thee}).

Results related to cycles with length constraints, with prescribed vertices, or to extensions of cycles are presented in the forthcoming tables.

\medskip

\noindent
\begin{tabular}{|p{2cm}|p{3.5cm}|p{3.5cm}|p{0.55cm}|p{3.95cm}|}
\hline
   Ref.& Guest class $\mathcal{H}$& Host class $\mathcal{G}$& T. &Up.-bound on the gap\\ \hline\hline
  \cite{erdHos1965independent}& cycles & any & \sv & $O(k\log k)$\\ \hline
  \cite{Simonovits67genera} & cycles & any& \sv & $\left (4 +
                                                  o(1)\right )k \log k$\\ \hline
\cite{voss1968some} & cycles & any & \sv & $\left (2 + o(1) \right )k \log k$\\\hline
  \cite{Diestel05grap} & cycles & any & \se & $(2 + o(1))k \log k$\\
  \hline
\cite{Ding2002381} & cycles & $\mathcal{G}_1$, weighted & \sf w &
                                                                       $
                                                                       k$\\\hline 
\cite{Ding2003244} & cycles & $\mathcal{G}_2$ & \sv & $k$\\\hline
\multirow{2}{*}{\cite{KloksLL02newa}} & \multirow{2}{*}{cycles} & planar graphs & \sv & $5k$\\\cline{3-5}
 & & outerplanar graphs & \sv & $2k$\\\hline
\multirow{2}{*}{\cite{ma2013approximate}} & \multirow{2}{*}{cycles} & \multirow{2}{*}{planar graphs} & \sv & $3k$\\\cline{4-5}
&&&\se & $4k-1$\\ \hline
\scalebox{0.95}{\cite{Reed96}} & directed cycles & any digraph& \sv & --\\
\hline
\cite{Reed1995gallai} & directed cycles & planar digraphs& \sv & $O(k\log(k)\log\log k)$\\
\hline
\multirow{2}{*}{\cite{Guenin2010packing}} & \multirow{2}{*}{directed cycles} & strongly planar digraphs& \multirow{2}{*}{\sv} & \multirow{2}{*}{$k$}\\
\cline{3-3}
&& $\mathcal{F}$ &&\\
\hline
\cite{lucchesi1978minimax} & directed cycles & planar digraphs & \se & $k$\\
\hline
\cite{Seymour1996Packing} & directed cycles & eulerian digraphs with a linkless embedding in 3-space & \se & $k$\\
\hline
\cite{2016arXiv160507082H} & cycles non homologuous to zero & embedded graphs & $\sv_{1/2}$ & -- \\\hline
\end{tabular}

\subsubsection{Cycles with length constraints}
The class of cycles (resp.\ directed cycles) of length at least $t$ is referred to as
$\mathcal{C}_{\geq t}$ (resp.\ $\vec{\mathcal{C}}_{\geq t}$).
 For every
positive integer $k$ with, we say that a graph is
\emph{$k$-near bipartite} if every set $X$ of vertices contains a
stable set of size at least~$|X|/2 - k.$

\medskip

\noindent
\begin{tabular}{|p{2cm}|p{3.5cm}|p{3.5cm}|p{0.55cm}|p{3.95cm}|}
\hline
   Ref.& Guest class $\mathcal{H}$& Host class $\mathcal{G}$&
                                                              T. &Up.-bound on the gap\\ \hline\hline
  \cite{reed1999mangoes} & odd cycles & planar graphs & \sv &
  superexponential\\ \hline
{\cite{fiorini2005approximate}} & {odd cycles} & {planar graphs} & \sv & $10k$\\\hline
  \cite{thomassen01conn} & odd cycles & $2^{3^{9k}}$-connected graphs & \sv & $2k-2$\\\hline
  \cite{rautenbach2001erdHos} & odd cycles & $576k$-connected graphs &
                                                                       \sv & $2k-2$\\\hline
\cite{Kawarabayashi2009} & odd cycles & $24k$-connected graphs & \sv
                                                                 &$2k-2$\\\hline
\cite{reed1999mangoes} & odd cycles & $k$-near bipartite graphs & \sv
                                                                 & --\\\hline
\cite{Kawarabayashi2007764} & odd cycles & embeddable in an orientable
                                           surface of Euler genus $t$ &
                                                                  \sv/\se
                                                               & --
  \\\hline
\cite{Berge2000197} & odd cycles & any & \se & -- \\ \hline
\cite{Kral2004107} & \multirow{2}{*}{odd cycles} & \multirow{2}{*}{planar graphs} & \multirow{2}{*}{\se} &
\multirow{2}{*}{$2k$}\\
\cite{fiorini2005approximate}&&&&\\\hline
  \cite{KawarabayashiK12} & odd cycles & 4-edge-connected graphs & \se &
                                                                      $2^{2^{O(k
                                                                      \log
                                                                      k)}}$
                                                                       \\ \hline

\cite{reed1999mangoes} & odd cycles & any & $\sv_{1/2}$ & --\\ \hline
  \cite{thomassen1988presence}  & cycles of length 0 mod $t$ & any &
                                                                       \sv & --\\\hline
\cite{Kawarabayashi2005271} & non-zero cycles & $(15k/2)$-connected
                                                group-labeled graphs &
                                                                      \sv
                                                           &
                                                             $2k-2$\\\hline
  \multirow{2}{*}{\cite{wollan2011packing}} & non-zero cycles & group-labeled graphs, c.f.\ \cite{wollan2011packing}& \multirow{2}{*}{\sv} & \multirow{2}{*}{$c^{k^{c'}}$ for some $c,c'$}\\\cline{2-3}
& cycles of non-zero length mod $2t+1$ & any & &\\\hline
\multirow{2}{*}{\cite{2016arXiv160507082H}} & doubly non-zero cycles, c.f.\ \cite{2016arXiv160507082H} & doubly group-labeled graphs & \multirow{2}{*}{$\sv_{1/2}$} & \multirow{2}{*}{--} \\\cline{2-3}
&odd cycles non homologuous to zero & embedded graphs & & \\\hline
\cite{BirmeleBR07} & $\mathcal{C}_{\geq t}$ & any digraph& \sv & $(13+o_t(1))tk^2$ \\ \hline
  \cite{JGT:JGT21776} & $\mathcal{C}_{\geq t}$ & any digraph& \sv &$(6t+4+o_t(1))k\log k$\\
                                                               \hline
{\small \cite{Mousset2016NSWtight}} &$\mathcal{C}_{\geq t}$ & any digraph& \sv & $6kt + (10 + o(1)) k \log k$\\\hline
\cite{havet:hal-00816135} & $\vec{\mathcal{C}}_{\geq 3}$ & any digraph&\sv&--\\
\hline
\end{tabular}

\subsubsection{Extensions of cycles}
A \emph{dumb-bell} is a graph obtained by
connecting two cycles by a (non-trivial) path.
\medskip

\noindent
\begin{tabular}{|p{2cm}|p{3.5cm}|p{3.5cm}|p{0.55cm}|p{3.95cm}|}
\hline
   Ref.& Guest class $\mathcal{H}$& Host class $\mathcal{G}$& T. &Up.-bound on the gap\\ \hline\hline
\cite{Simonovits67genera} & dumb-bells & any& \sv & $(4000+o(1))k \log k$\\ \hline
  \cite{FominLMPS13quad} & $\mathcal{M}(\theta_t)$& any & \sv & $O(t^2k^2)$\\
  \hline
  \cite{FioriniJS13} &  $\mathcal{M}(\theta_t)$& any & \sv & $O_t(k \log k)$\\
\hline
  \multirow{2}{*}{\cite{Raymond2013edge}} &\multirow{2}{*}{$\mathcal{M}(\theta_t)$} & \multirow{2}{*}{any} & \multirow{2}{*}{\se} & $O(k^2t^2 \polylog kt)$\\
      &&&& $O(k^4t^2 \polylog kt)$ \\
\hline
  \parbox{2cm}{\cite{Chatzidimitriou2015logopt},\\
  {\hyperref[thetaEP]{Cor.\ \ref*{thetaEP}}}\vspace{1mm}}& $\mathcal{M}(\theta_t)$& any & {\sf
                                                               v}/{\sf
                                                               e} &
                                                                    $O_t(k
                                                                    \log
                                                                    k)$\\
                                                                    \hline
\end{tabular}

\subsubsection{Minor models}
For every digraph $D$, we denote by $\vec{\mathcal{M}}(D)$
(respectively $\vec{\mathcal{T}}(G)$, $\vec{\mathcal{I}}(G)$) the class of all digraphs that
contain $D$ as a directed minor (respectively directed topological minor, directed immersion). Refer to \cite{chudnovsky2011well, Chudnovsky2012tourn, Fradkin2013tourn} for a definition of these notions.

We also denote by $\vec{\mathcal{M}}_{\bowtie}(D)$
(respectively $\vec{\mathcal{T}}_{\bowtie}(G)$) the class of all digraphs that
contain $D$ as a butterfly-minor (respectively as a butterfly
topological minor). $\vec{\mathcal{P}}$ (respectively
$\vec{\mathcal{W}}$) is the class of all graphs that are butterfly
minors of a cylindrical directed grid (respectively butterfly
topological minors of a cylindrical directed wall). See for
instance~\cite{AkhoondianKKW09thee} for a definition of the
cyclindrical directed grid and wall and~\cite{JohnsonRST01dire,AkhoondianKKW09thee} for a definition of butterfly (topological) minors.

For every $s\in \N$, a digraph is said to be \emph{$s$-semicomplete} if 
for every vertex $v$ there are at most $s$ vertices that are not connected to $v$ by an arc (in either direction). A \emph{semicomplete} digraph is a 0-semicomplete digraph.
\medskip

\noindent
\begin{tabular}{|p{2cm}|p{3.5cm}|p{3.5cm}|p{0.55cm}|p{3.95cm}|}
\hline
   Ref.& Guest class $\mathcal{H}$& Host class $\mathcal{G}$& T. &Up.-bound on the gap\\ \hline\hline
  \multirow{2}{*}{\parbox{2cm}{\cite{RobertsonS86GMV}, \\ \autoref{ep-tree}}} & \multirow{2}{*}{$\mathcal{M}(H)$, $H$ planar} & any & \sv &
                                                                      --\\ \cline{3-5}
& & $\{G,\ \tw(G)\leq t\}$ & \sv & $(t-1)(k\cc(H)-1)$\\\hline
  \cite{2010arXiv1003.3915D} & $\mathcal{M}(K_t)$ & $O(kt)$-connected
                                                    graphs & \sv &-- \\ \hline
 \cite{FominST11stre} & $\mathcal{M}(H)$, $H$ planar connected & $K_{t}$-minor free & \sv & $O_{H,t}(k)$\\ \hline
  \cite{RaymondT13poly} & $\mathcal{M}(H),\ \pw(H)\!\leq\! 2$ and $H$ connected & any & \sv
                                                           &
                                                             $2^{O(|H|²)}\cdot k²\log k$\\ \hline
  \cite{ChekuriC13Larg}+\newline \cite{ChekuriC13poly}, \coref{epconn} & $\mathcal{M}(H)$, $H$ planar connected & any & \sv & $O(|H|^{O(1)}\cdot k \polylog k)$\\ \hline
 \multirow{2}{*}{\parbox{2cm}{\cite{Chatzidimitriou2015logopt}, \coref{thetaEP}}} & $\mathcal{M}(H)$, $H$ connected & $\{G,\ \tpw(G) \leq t\}$  & \sv/\se & $O_{H,t}(k)$\\
\cline{2-5}
& $\mathcal{M}(\theta_{t,t'})$ & simple graphs & \se & --\\
\hline
\scalebox{0.95}{\cite{AkhoondianKKW09thee}} & $\vec{\mathcal{M}}_{\bowtie}(H),\ H \in \vec{\mathcal{P}}$ & any digraph& \sv & --\\
\hline
\end{tabular}

\subsubsection{Subdivisions}
For every $t\in \N$, $\mathcal{T}_{(0 \!\!\mod t)}(H)$ denotes the class of 
subdivisions of $H$ where every edge
is subdivided $0 \!\!\mod t$ times.
$\mathcal{L}$ is a
graph class defined in the (unpublished) manuscript~\cite{Liu2014topo}.
See the previous section for the definition of $\vec{\mathcal{T}}(G)$ and $\vec{\mathcal{W}}$.
\medskip

\noindent
\begin{tabular}{|p{2cm}|p{3.5cm}|p{3.5cm}|p{0.55cm}|p{3.95cm}|}
\hline
   Ref.& Guest class $\mathcal{H}$& Host class $\mathcal{G}$& T. &Up.-bound on the gap\\ \hline\hline
  \cite{thomassen1988presence}  & $\mathcal{T}_{(0 \!\!\mod t)}(H)$, $H$ planar subcubic & any &
                                                                       \sv & --\\
\hline
\cite{Liu2014topo} & $\mathcal{T}(H), H \in \mathcal{L}$ & any & \sv & --\\\hline
\scalebox{0.95}{\cite{AkhoondianKKW09thee}} & $\vec{\mathcal{T}}_{\bowtie}(H),\ H \in \vec{\mathcal{W}}$ & any digraph& \sv & --\\\hline
\end{tabular}

\subsubsection{Immersion expansions}
A graph $H$ is a {\em half-integral immersion} of a graph $G$
is $H$ is an immersion of the graph obtained by $G$  after duplicating 
the multiplicity of all its edges. We denote by 
$\mathcal{I}_{1/2}(H)$ the class  of all graphs containing $H$ as a 
half-integral immersion.
See above the definition of $\vec{\mathcal{I}}(G)$.
\medskip

\noindent
\begin{tabular}{|p{2cm}|p{3.5cm}|p{3.5cm}|p{0.55cm}|p{3.95cm}|}
  \hline
  Ref.& Guest class $\mathcal{H}$& Host class $\mathcal{G}$& T. &Up.-bound on the gap\\ \hline\hline
\cite{Liu2015EP} & $\mathcal{I}(H)$ & 4-edge-connected & \se & --\\ \hline
\multirow{3}{*}{\parbox{2cm}{\cite{Giannopoulou2016packing},
  \autoref{speimm} \coref{corrsplit}}} & $\mathcal{I}(H)$, $H$ planar subcubic
                                 connected non-trivial& any & \se
                                                                          & \parbox[t]{3cm}{$(\|H\| \cdot k)^{O(1)}$}\\
\cline{2-5}
& \multirow{2}{*}{\parbox{3.5cm}{$\mathcal{I}(H)$, $H$ connected non-trivial}} & $\{G, \tpw(G) \leq t\}$
                                                            & \multirow{2}{*}{\se} &
                                                                    \multirow{2}{*}{$\|H\|\cdot
                                                                    t^2 \cdot
                                                                    k$}\\
  \cline{3-3}
&&$\{G, \tcw(G) \leq t\}$ & & \\ \hline
\cite{Liu2015EP} & $\mathcal{I}_{1/2}(H)$ & any & $\se_{1/2}$ & --\\ \hline
\end{tabular}

\subsubsection{Patterns with prescribed vertices}

Let us first present the two settings of \ep\ problems with prescribed
vertices that we want to deal with here.
The first type is when the guest class consists of fixed subgraphs of the
host graph. For instance, one can consider a family $\mathcal{F}$ of (non
necessarily disjoint) subtrees of a tree $T$, and compare the maximum
number of disjoint elements in $\mathcal{F}$ with the minimum number
of vertices/edges of $T$ meeting all elements of~$\mathcal{F}.$ We will
refer to these guest classes by words indicating that we are dealing
with substructures (like ``subtrees''). We stress that in this
setting, the host class is allowed to contain one subgraph $F$ of the
host graph, but not one other subgraph $F'$ even if $F$ and $F'$ are isomorphic.
For every positive integer $t$, a $t$-path is a disjoint union of $t$
paths, and a $t$-subpath of a $t$-path $G$ is a subgraph that has a
connected intersection with every connected component of $G.$ The
concept of $t$-forests and $t$-subforests is defined similarly.

In order to introduce the second type of problem, we need the
following definition. Let $\sx \in \{\sv, \se\}.$ If $\mathcal{H}$ is a class of graphs, $G$ is a
graph and $S\subseteq A_\sx(G)$, then a $S$-$\mathcal{H}$-subgraph of
$G$ is a subgraph of $G$ isomorphic to some member of $\mathcal{H}$
and that contain one edge/vertex of $S.$
We are now interested in comparing, for every graph $G$ and every
$S\subseteq A_\sx(G)$, the maximum number of $S$-$\mathcal{H}$-subgraph of
$G$ with the minimum number of elements of $A_\sx(G)$ that meet all
$S$-$\mathcal{H}$-subgraphs of~$G.$ We refer to these problems by
prefixing the guest class with an~``$S$'' (like in ``$S$-cycles'').
The authors of \cite{2016arXiv160507082H} consider $(S_1,S_2)$-cycles for $S_1,S_2 \subseteq V(G)$: such cycles must meet both of $S$ and~$S'$.
A generalization of this type of problem has been introduced in
\cite{Kwon2015ep}: instead of one set $S$, one considers three subsets $S_1, S_2, S_3$ of $V(G)$ and a $(S_1,S_2,S_3)$-subgraph is required to intersect at least two sets of~$S_1$, $S_2$ and $S_3$.
Note that some results on patterns with prescribed vertices have been stated in the table on acyclic patterns.
Recall that $\cc(G)$ denotes the number of connected components of the graph~$G$.
\medskip

\noindent
\begin{tabular}{|p{2cm}|p{3.5cm}|p{3.5cm}|p{0.55cm}|p{3.95cm}|}
\hline
   Ref.& Guest class $\mathcal{H}$& Host class $\mathcal{G}$& T. &Up.-bound on the gap\\ \hline\hline
  \cite{hajnal1958auflosung} & subpaths & paths & \sv & $k$\\ \hline
\multirow{3}{*}{\cite{gyarfas1970helly}}& $t$-subpaths & $t$-paths & \sv & $O(k^{t!})$\\ \cline{2-5}
& subgraphs $H$ with $\cc(H) \leq t$ & paths & \sv &--\\ \cline{2-5}
& $t$-subforests & $t$-forests& \sv & --\\ \hline
\cite{gyarfas1970helly} & subtrees of a tree & trees & \sv & k\\ \hline
\cite{Kaiser}& $t$-subpaths & $t$-paths & \sv & $(t^2-t+1)k$\\ \hline
\cite{alon98piercing}& $t$-subpaths & $t$-paths & \sv & $2t^2k$\\ \hline
\multirow{3}{*}{\cite{Alon02cove}} & subgraphs $H$ with $\cc(H) \leq t$ & trees & \sv & $2t^2k$\\ \cline{2-5}
& subgraphs $H$ with $\cc(H) \leq t$ & $\{G,\ \tw(G)\leq w\}$ & \sv & $2(w+1)t^2k$\\ \hline  
\cite{Kakimura2011378} & $S$-cycles & any & \sv & $O(k^2 \log k)$
                                                                        \\
  \hline \cite{Pontecorvi20121134} & $S$-cycles & any &
                                                                       \sv/\se &
                                                              $O(k \log k)$ \\ \hline
  \cite{2014arXiv1412.2894B} & $S\text{-cycles} \cap \mathcal{C}_{\geq t}$ & any & \sv &
                                                               $O(tk\log
                                                                                  k)$ \\ \hline
\cite{2014arXiv1411.6554J} & odd $S$-cycles & $50k$-connected graphs & \sv & $O(k)$\\\hline 
\cite{Kakimura2013} & odd $S$-cycles & any & $\sv_{1/2}$ & -- \\\hline
\cite{doi:10.1137/100786423} & directed $S$-cycles & all diraphs & $\sv_{1/5}$ & --\\hline
\cite{kawarabayashi2013packing} & odd directed $S$-cycles & any digraph & $\sv_{1/2}$ & -- \\\hline
\cite{2016arXiv160507082H} & $(S_1,S_2)$-cycles & any & $\sv$ & -- \\\hline
\cite{Kwon2015ep} & $(S_1,S_2,S_3)$-$\mathcal{M}(H)$, $H$ planar & any & \sv & --\\
\hline 
\end{tabular}

\subsubsection{Classes with bounded parameters}
\noindent
\begin{tabular}{|p{2cm}|p{3.5cm}|p{3.5cm}|p{0.55cm}|p{3.95cm}|}
  \hline
Ref.& Guest class $\mathcal{H}$& Host class $\mathcal{G}$&
                                                              T. &Up.-bound on the gap \\ \hline\hline
\cite{thomassen1988presence} & any family of connected graphs & $\{G,\ \tw(G)\leq t\}$ & \sv & $k(t+1)$\\\hline
\cite{FioriniJW12excl}   & $\{H,\ \pw(H) \geq t\}$ & any & \sv & $O_t(k)$\\ \hline
\cite{Chatzidimitriou2015logopt} & any finite family of connected graphs & $\{G,\ \tpw(G) \leq t\}$ & \sv/\se &  $O_{t}(k)$\\
\hline
\end{tabular}

\subsection{Negative results}
\label{sec:neg}
The next table presents lower bounds on the gap for several graph
classes, as well as graph classes that do not have the \ep\ 
property. It indicates
to which extend the results of the table of \autoref{posires} are best possible.
The notation used here are the same as in the previous section, where
they are defined.

\subsubsection{Cycles and paths}

\noindent
\begin{tabular}{|p{2cm}|p{3.5cm}|p{3.5cm}|p{0.55cm}|p{3.95cm}|}
  \hline
Ref.& Guest class $\mathcal{H}$& Host class $\mathcal{G}$&
                                                              T. &Gap 
  \\ \hline\hline
\cite{Tuza90aconj} & triangles & all graphs & \se & $\geq 2k$\\\hline
\cite{erdHos1965independent} & cycles & all graphs & \sv & $\Omega(k \log
                                                           k)$\\\hline
  \cite{Simonovits67genera} & cycles & all graphs& \sv & $> \left (\frac{1}{2}
                                                         + o(1) \right ) k \log k$\vspace{1mm}\\ \hline
\cite{voss1968some} & cycles & all graphs & \sv & $\geq \left (\frac{1}{8} + o(1) \right )k \log k$\vspace{1mm}\\\hline
\cite{KloksLL02newa} & cycles & planar graphs & \sv & $\geq2k$\\ \hline
\cite{{ma2013approximate}} & cycles & planar graphs & \se & \parbox{2.4cm}{$\geq
                                                            4k-c$,\\ $c\in \N$\vspace{1mm}}\\\hline
\cite{dejter1988unboundedness} & odd cycles & all graphs & \sv & arbitrary\\
  \hline
\cite{thomassen1988presence}  & cycles of length\newline $p \!\!\mod t$, $p \in \intv{1}{t-1}$ & all graphs &
                                                                       \sv & arbitrary\\
\hline
\cite{reed1999mangoes} & odd cycles & all graphs & \se & arbitrary\\ \hline
\cite{thomassen01conn} & odd cycles & planar graphs & \sv & $\geq 2k-2$\\ \hline
\cite{Kral2004107} & odd cycles & planar graphs & \se & $\geq 2k$ \\ \hline
   \cite{Pontecorvi20121134} & $S$-cycles & any &
                                                                       \sv &
                                                              $\Omega(k
                                                                             \log k)$ \\ \hline
\cite{Kakimura2013} & odd $S$-cycles & all graphs & \sv & arbitrary \\\hline
\cite{doi:10.1137/100786423} & directed $S$-cycles & all diraphs & \sv/\se & arbitrary\\\hline
\scalebox{0.95}{\cite{kawarabayashi2013packing}} & odd directed\newline $S$-cycles & all digraphs & \sv & arbitrary \\\hline
\multirow{2}{*}{\cite{JGT:JGT21776}} & \multirow{2}{*}{$\mathcal{C}_{\geq t}$} & \multirow{2}{*}{all graphs} & \multirow{2}{*}{\sv} &
                                                                  \parbox{2.4cm}{$\Omega(k\log
                                                                  k)$,\\
                                                                  $t$
                                                                  fixed\vspace{1mm}}\\\cline{5-5}
&&&&$\Omega(t)$, $k$ fixed\\ \hline
\multirow{2}{*}{\cite{Mousset2016NSWtight}} & \multirow{2}{*}{$\mathcal{C}_{\geq t}$} & \multirow{2}{*}{all graphs} & \multirow{2}{*}{\sv} & $\geq(k-1)t$\\\cline{5-5}
&&&& $\geq \frac{(k-1)\log k}{8}$\\\hline
\cite{Simonovits67genera} & dumb-bells & all graphs& \sv & $>(1+o(1))k \log k$\\ \hline
\cite{marx2015exact} & $H$-valid paths,\newline $H$ with no\newline matching of size $t$ & all graphs & \sv & unavoidable dependency in $t$\\\hline
\end{tabular}

\subsubsection{Patterns related to containment relations}

\noindent
\begin{tabular}{|p{2cm}|p{5cm}|p{3.5cm}|p{0.55cm}|p{2.5cm}|}
  \hline
Ref.& Guest class $\mathcal{H}$& Host class $\mathcal{G}$&
                                                              T. &Gap
  \\ \hline\hline
{from \cite{erdHos1965independent}} & $\mathcal{M}(H)$, $H$ has a cycle & all graphs & \sv & $\Omega(k \log
                                                           k)$\\\hline
\cite{RobertsonS86GMV} & $\mathcal{M}(H)$, $H$ non-planar & all graphs
                                                         & \sv & arbitrary\\
  \hline
\autoref{neg-min-e-plus} & $\mathcal{M}(H)$, $H$ non-planar & all graphs
                                                         & \se & arbitrary\\
  \hline
\autoref{nega} & $\mathcal{T}(H)$, $H$ non-planar & all graphs
                                                         & \sv & arbitrary\\\hline
\cite{thomassen1988presence}  & $\mathcal{T}_{(p \!\!\mod t)}(H)$, $H$ planar subcubic, $p \in \intv{1}{t-1}$ & all graphs &
                                                                       \sv & arbitrary\\
\hline
 \cite{thomassen1988presence} & $\mathcal{T}(H)$, for infinitely many
                                trees $H$ with $\Delta(H)=4$ & planar graphs & \se & arbitrary\\
  \hline
\autoref{negae} & $\mathcal{T}(H)$, $H$ non-planar subcubic& all graphs
                                                         & \se & arbitrary\\
  \hline
 copying \cite{thomassen1988presence} & $\mathcal{I}(H)$, for infinitely many
                                trees $H$ with $\Delta(H)=4$ & planar graphs & \se & arbitrary\\
  \hline
\coref{neg-min} & $\mathcal{I}(H)$, $H$ non-planar & all graphs
                                                         & \sv & arbitrary\\\hline
\coref{neg-min-e} & $\mathcal{I}(H)$, $H$ non-planar subcubic& all graphs
                                                         & \se & arbitrary\\
  \hline
                                                         
  \scalebox{0.95}{\cite{Giannopoulou2016packing}} & $\mathcal{I}(H)$, for some 3-connected $H$
                           with $\Delta(H)=4$ & planar graphs & \se &
                                                                      arbitrary\\
  \hline
\cite{Liu2015EP} & $\mathcal{I}(H)$, for every $H$ & 3-edge-connected graphs & \se & arbitrary\\\hline
\multirow{2}{*}{\cite{AkhoondianKKW09thee}} & $\vec{\mathcal{M}}(G),\ G \not \in \vec{\mathcal{P}}$ &\multirow{2}{*}{all digraphs} & \multirow{2}{*}{\sv} & \multirow{2}{*}{arbitrary}\\
\cline{2-2}
 & $\vec{\mathcal{T}}(G),\ G \not \in \vec{\mathcal{W}}$&&&\\
\hline
\end{tabular}

\subsection{Some questions and conjectures}

Clearly, the most general question on the
 \ep\ property is to characterize the class $\mathcal{EP}^{\sx}_{\lleq}$ (defined 
in \autoref{sec:cr}) for various instantiations of $\sx$ and $\lleq$ and optimize the corresponding gap.
In what follows we sample some related conjectures and questions that have appeared in the bibliography.

\begin{question}[\!\!\cite{thomassen1988presence}]
  Is it true that for every class $\mathcal{H}$ of graphs, either
  $\mathcal{H}$ has the \sv-\ep\ property or there is no integer $q$ such for every graph $G$ with $\vpack_{\mathcal{H}}(G)\leq 1$ it holds that $\vcover_{\mathcal{H}}(G)\leq q.$
In particular, it is true when $\mathcal{H}$
  consists of connected graphs and is closed under topological minors?
\end{question}

\begin{conjecture}[Tuza's conjecture~\cite{Tuza90aconj}]
  For every graph $G$ it holds that
\[
  \ecover_{\{K_3\}}(G) \leq 2\cdot  \epack_{\{K_3\}}(G).
\]
\end{conjecture}

\begin{conjecture}[\!\!\cite{BirmeleBR07}]\label{cover2}
Let $l\geq 6$ be an integer. Let $G$ be a graph containing no \sv-$\mathcal{C}_{\geq l}$-packing of
size 2. Then there exists a \sv-$\mathcal{C}_{\geq l}$-cover of $G$ of
size at most~$l.$
\end{conjecture}

\begin{conjecture}[Jones' conjecture~\cite{KloksLL02newa}]
Let $\mathcal{C}$ denote the class of all cycles.
  For every planar graph $G$, it holds that \[\vcover_{\mathcal{C}}(G)
  \leq 2 \cdot \vpack_{\mathcal{C}}(G).\]
\end{conjecture}

A \emph{hole} is an induced cycle of lenght at least~4.
\begin{question}[\cite{jansen2016approximation}]
  Is there a function $f\colon \N \to \N$ such that for every graph
  $G$ and every $k\in \N$, the following holds:
  \begin{itemize}
  \item $G$ has $k$ vertex-disjoint holes; or
  \item there is a set $X\subseteq V(G)$ such that $G\setminus X$ has
    no hole?
  \end{itemize}
\end{question}




\section*{Acknowledgements}

The authors wish to thank Dimitris Chatzidimitriou, Fedor V. Fomin, Gwenaël Joret, Daniel Lokshtanov, Ignasi Sau, and Saket Saurabh
for sharing their ideas and viewpoints on the \ep\ property, as well as the anonymous referees for their useful comments.

\medskip

\newcommand{\etalchar}[1]{$^{#1}$}

\end{document}